\documentclass[journal]{IEEEtran}
\ifdefined\pdfobjcompresslevel
\fi

\usepackage{dblfloatfix}
\usepackage[section]{placeins}
\usepackage{amsmath,amsfonts}
\usepackage{amsthm}
\newtheorem{proposition}{Proposition}
\newtheorem{theorem}{Theorem}
\usepackage{graphicx}
\usepackage{tikz}
\usepackage[utf8]{inputenc}   
\usepackage[T1]{fontenc}      
\usepackage{lmodern}          
\usetikzlibrary{quantikz2,backgrounds,fit,decorations.pathreplacing}
\usepackage{xcolor}%
\usepackage{cite}
\usepackage{url}
\usepackage{algorithm}
\usepackage{algorithmic}
\usepackage{soul}
\usepackage{hyperref} 
\usepackage{amssymb}

\begin{document} 

\title{Optimal Entanglement Routing in Quantum Repeater Chains: Beyond Fixed Operation Order and Purification Schedule}

\author{
Aikaterini Mandilara,
Antonia Tsili,
Dimitris Syvridis,
Konstantinos Christodoulopoulos
\thanks{A. Mandilara, D. Syvridis and K. Christodoulopoulos are with the Dep. of Informatics and Telecommunications, National and Kapodistrian University of Athens (e-mail: mandkat@di.uoa.gr). A. Tsili and A. Mandilara is also with Eulambia Advanced Technologies Ltd. }
}

\maketitle

\begin{abstract}
Entanglement routing establishes entangled pairs between distant nodes of a quantum network by purifying and swapping pairs generated on elementary links. Existing methods typically restrict the decision space along two axes: the operation order, often fixed to purify-then-swap (PtS), and the purification schedule of each link, often restricted to pumping. Focusing on linear repeater chains with finite link capacities, we relax both restrictions and study the maximization of the expected end-to-end throughput subject to a fidelity threshold under two quantum-noise models. We develop a unified capacity-aware optimization framework comprising an exact mixed-integer linear program (MILP) under PtS, instantiated with either pumping or general tree purification schedules, and an exact dynamic program (DP) over arbitrary operation orders. Under symmetric Pauli noise, we prove that pumping converges to a fidelity strictly below unity, imposing a capacity-independent feasibility bound and a maximum chain length on every pumping-based PtS method, whereas general tree schedules yield a capacity-dependent bound. We further prove that post-swap purification never increases the throughput, so a free operation order can only increase feasibility. Numerical evaluations confirm the bounds: tree schedules serve over 90\% of the requests that pumping cannot at moderate fidelity thresholds; within the pumping class, a free operation order recovers much of this advantage; but once links use tree schedules, the tree-based MILP and the order-exact DP serve identical request sets throughout. Operation order and purification schedule are thus substitutes, with the purification schedule the dominant factor governing feasibility under symmetric Pauli noise.
\end{abstract}


\section{Introduction}\label{sec:introduction}

In his seminal paper, J. Kimble \cite{Kimble2008QuantumInternet} introduced the quantum Internet as a framework that ``offers a unifying set of opportunities and challenges across exciting intellectual and technical frontiers, including quantum computation, communication, and metrology''. Since then, quantum networking has made significant progress \cite{Illiano2022QIPStack,Azuma2022QuantumRepeaters} toward practical realization, while introducing new challenges across quantum information science, communication engineering, and network optimization \cite{Wehner2018QuantumInternet}.

A quantum network distributes maximally entangled Bell-state photon pairs (often referred to as Einstein--Podolsky--Rosen, or EPR, pairs), enabling a wide range of applications such as quantum teleportation, remote gates for distributed quantum computing, quantum-secure communication, and quantum-enhanced sensing and metrology. Entangled pairs are the resource these applications consume, but their distribution over long distances is limited by photon loss and channel noise, which is why quantum repeaters are needed. While loss destroys a pair outright, channel imperfections degrade the quality of the pairs that survive. Throughout this work we quantify that quality by the \textit{fidelity} with respect to the target Bell state; for the state families considered here (Section~\ref{sec:noise_models}) this single scalar characterizes the state completely.

The quantum network builds upon the quantum \emph{elementary links} that are established between \emph{adjacent} nodes through entangled-pair generation, transfer, and local memory storage. The problem studied here begins once the pairs of the elementary links are in place. To extend communication over multiple hops, quantum repeaters perform entanglement \textit{swapping}, an operation that consumes two entangled pairs spanning adjacent links to establish a new pair across a \emph{virtual link} between their distant endpoints. While swapping extends communication distance, it degrades the fidelity (the resulting pair has lower fidelity than either input pair). To counteract this, repeaters apply entanglement \textit{purification}, which probabilistically combines two lower-fidelity pairs into a single pair of higher fidelity. Distance thus costs fidelity, and fidelity costs pairs. The pairs available on each elementary link are finite ---bounded by the generation rate and the memory size--- and perishable, since decoherence degrades stored pairs over time. In this setting, \emph{capacity-aware entanglement routing} decides how the finite pairs of each link are purified and swapped to establish end-to-end entanglement, trading the rate of delivered pairs against their fidelity. In this work we formalize the objective as maximizing the expected end-to-end throughput (rate of delivered entangled pairs) subject to a fidelity threshold.

Determining when and where purification and swapping operations should be applied across multi-hop paths, while accounting for hardware and physical constraints, makes entanglement routing a challenging combinatorial optimization problem. Existing approaches typically restrict the decision space along two axes. First, they fix the \emph{operation order} to \emph{purify-then-swap (PtS)}, where entangled pairs on elementary links are purified before any swapping occurs. In general, purification and swapping may be interleaved in any order consistent with protocol's dependencies. Second, they restrict link-level \emph{purification schedule} to \textit{pumping}, in which a single pair is iteratively purified against fresh raw pairs, so that the schedule is fully specified by its purification "depth". In general, however, purification can follow arbitrary tree schedules, where both input pairs to a purification step may be outputs of prior purification operations. These two restrictions are distinct architectural choices: \textit{(i)} the \emph{operation order} of purification and swapping, and \textit{(ii)} the \emph{purification schedule}. Relaxing these opens two optimization dimensions that we explore in this paper.

Underlying both axes is a physical limitation: the physical noise affecting the entangled pairs determines the fidelity that purification can attain, and, consequently, changes the role of purification in the route constructions. We account for this by considering two Bell-diagonal noise models with markedly different purification behavior: a symmetric Pauli noise model, leading to Werner states, in which errors populate all three undesired Bell components, and a single-Pauli-error noise model, leading to two-Bell-state diagonal states, in which the noise is confined to a single error channel. 

\begin{figure}[b]
    \centering
    \includegraphics[width=\linewidth]{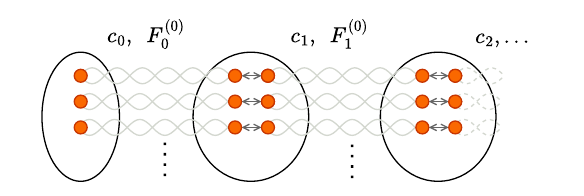}
   \caption{A schematic representation of the linear quantum repeater chain considered in this work. Neighboring repeater nodes are connected through elementary links $e$, each characterized by a baseline fidelity $F_e^{(0)}$ and an entanglement-generation capacity $c_e$. Long-distance entanglement is established through sequences of entanglement purification and entanglement swapping operations. Curly lines denote shared entanglement, while red dots represent matter qubits stored in quantum memories.}
    \label{fig:1}
\end{figure}

In this work, we investigate entanglement routing on a linear quantum repeater chain~\cite{dur1999_quantum_repeaters} (Fig.~\ref{fig:1}), a fundamental building block of more complex networks, where each elementary link $e$ is characterized by a finite entanglement capacity $c_e$,  and a baseline fidelity $F_e^{(0)}$ degraded by noise and memory decoherence. We develop an optimization framework that maximizes the expected end-to-end throughput subject to a fidelity constraint while systematically exploring the two design axes identified above: purification schedule and operation order. Our main contributions are:
\begin{itemize}
\item \textbf{A unified optimization framework.} We formulate a mixed-integer linear programming (MILP) model that is exact under PtS operation order and can be instantiated with either pumping or general tree purification schedules. We also develop a low-complexity greedy heuristic for online operation. We further develop a dynamic programming (DP) formulation that is exact over arbitrary purification and swapping operation orders for a given schedule class. Instantiated with tree schedule frontiers, the DP thus provides the global optimum over both schedule and order dimensions considered in this work.
\item \textbf{Analytical feasibility bounds.} Under symmetric Pauli noise, we prove that pumping converges to a closed-form fixed point below unit fidelity, yielding a capacity-independent end-to-end fidelity ceiling and a maximum chain length for every pumping-based PtS method. General purification tree schedules escape this ceiling, leading, instead, to a capacity-dependent feasibility bound.
\item \textbf{A separation between schedule and order.} We prove that post-swap purification can never increase throughput, as it consumes pairs' width and success probability to gain fidelity. Thus, a free operation order can add feasibility, or reach the threshold with cheaper link schedules, but cannot raise the throughput of any routing solution. This separates the two axes: the operation order adds value only where the link schedules alone cannot reach the fidelity threshold, and once they can, its benefit is already captured by a good schedule.
\item \textbf{Noise-model and objective generality.} We show that our formulations depend only on three structural properties---diminishing purification returns (used only by the heuristic), log-linear swapping, and an objective that is log-linear for the MILP and decomposable over sub-paths for the DP---rather than on specific state equations, allowing both Werner and two-Bell-state diagonal models to be solved without structural re-derivation.
\end{itemize}


We evaluated the framework through simulations on repeater chains under both noise models. Measured blocking matched the analytical bounds, and tree schedules served over 90\% of the requests infeasible under pumping at moderate fidelity thresholds. The tree-based MILP and the order-exact DP served identical instance sets throughout, the DP retaining only a modest throughput advantage at longer chains. Within the pumping class, an arbitrary operation order recovered much of the trees schedule advantage by purifying after swapping, i.e., by accessing a purification recursion unavailable to pumping on individual links. Q-PATH~\cite{li2022_qpath_qleap} blocked exactly as the pumping bound predicts, with throughput lower by construction, and the per-link optimal schedules of Chen and Jia~\cite{chen2024_jsac} lost to chain-wide coordination on heterogeneous links, while the greedy heuristic reproduced the optimal served set at millisecond runtimes.

The rest of the paper is organized as follows. Section~\ref{sec:related} reviews the relevant literature. Section~\ref{sec:ent_models_and_protocols} introduces the noise models and entanglement protocols considered in this work. Section~\ref{sec:schedules} examines purification schedule and operation order and derives the analytical results that guide the subsequent optimization. Section~\ref{sec:operation_model} formalizes the network operation model, including link capacities, memory decoherence, and timing constraints, while Section~\ref{sec:algorithms} formulates the capacity-aware routing problem and presents the proposed algorithms. The framework is evaluated numerically in Section~\ref{sec:results}, and Section~\ref{sec:conclusions} concludes the paper. Supporting theoretical results and proofs are provided in the Appendices.


\section{Related Work}
\label{sec:related}
While point-to-point quantum primitives have been studied extensively, optimizing multi-hop entanglement distribution at the network level is a comparatively young research direction. As quantum networks scale, classical networking problems such as routing, resource allocation, and throughput optimization re-emerge under distinctly quantum constraints~\cite{Shi2020Concurrent}. Entanglement routing therefore lies at a natural intersection of quantum information and classical network optimization, bringing together quantum communication protocols, communications engineering, operations research, and algorithm design.

Recent papers target general topologies and multi-commodity flows~\cite{wang2026_sftrap, xiao2024_psc, zhang2025_link_config, nguyen2025_merr}. A fidelity-aware routing algorithm over such a topology must decide, for each candidate path, how each link is purified and in which order purification and swapping operations are applied. We isolate this path-level subproblem and solve it exactly on the linear quantum repeater chain, the fundamental building block of more complex architectures~\cite{dur1999_quantum_repeaters, inesta2023_homog_chains}. The resulting path-level framework can naturally serve as a building block for future extensions incorporating path selection and multi-commodity resource contention.

However, the typical formulation of the entanglement-routing problem remains the maximization of end-to-end throughput subject to a fidelity constraint. Most existing approaches are heuristic or learning-based~\cite{wang2026_sftrap, xiao2024_psc, li2022_qpath_qleap, ni2025_joint}, trading optimality for scalability to large topologies and multi-request scheduling. Among them, Q-PATH and Q-LEAP~\cite{li2022_qpath_qleap} are representative approaches; we adopt Q-PATH as a baseline in Section~\ref{sec:results}. Several works have also pursued exact optimization. Zhang \emph{et al.}~\cite{zhang2025_link_config} formulate joint link configuration and purification as a MILP, employing a logarithmic transformation of the multiplicative Werner-parameter composition. Nguyen \emph{et al.}~\cite{nguyen2025_merr} formulate multi-commodity entanglement routing as an ILP with LP-relaxation rounding, while Zhang \emph{et al.}~\cite{zhang2025_satellite} apply Benders' decomposition to satellite quantum networks. Kar and Mukhopadhyay~\cite{kar2026_utility} propose a mixed-integer convex programming framework based on logarithmic utility functions, while Chen and Jia~\cite{chen2024_jsac} develop an exact DP for purification schedule and derive conditions under which PtS is optimal on multi-hop paths.

Despite their methodological differences, these exact formulations restrict the operation order and, with the exception of~\cite{chen2024_jsac}, further restrict link purification to pumping through the selection of a per-link purification depth. Several studies have questioned the restriction to a fixed operation order. Victora \emph{et al.}~\cite{victora2023_purification} show that asynchronous decoherence can render strict PtS inefficient and advocate dynamic selection among PtS, swap-then-purify (StP), and no purification; Haldar \emph{et al.}~\cite{haldar2025_quasilocal} report that StP can outperform PtS in multiplexed repeater chains at moderate-to-high fidelities; Zang \emph{et al.}~\cite{zang2023_buffer} identify buffer times as a decisive factor; and Koutsopoulos~\cite{koutsopoulos2024_milcom} formalizes the trade-off on two-link chains. Vecino Pe{\~n}as \emph{et al.}~\cite{vecino2026_strategy} optimize by dynamic programming the purification rounds applied to every pair of a path, elementary and multi-hop alike, but with deterministic operations and fidelity as the objective, so neither width nor throughput enters. What remains unavailable is an exact method that optimizes \emph{over} the operation order under probabilistic operations and finite capacity, where the order trades feasibility against throughput. Our DP fills this gap, allowing us to quantify separately the roles of operation order and purification schedule.

Concerning physical-layer noise, routing studies typically adopt either the single-Pauli-error noise model (two-Bell-state diagonal states)~\cite{li2022_qpath_qleap, wang2026_sftrap, koutsopoulos2024_milcom} or the symmetric Pauli noise model (Werner states)~\cite{xiao2024_psc, zhang2025_link_config}, with optimization formulations tailored to the chosen model. In contrast, our MILP and DP rely only on structural properties of the noise model and the objective, not on their specific equations, allowing the same formulations to be used without modification through model-dependent precomputed coefficients. These properties also identify a broader class of compatible noise models and objectives, as detailed in Appendix~\ref{app:convexity}. We next introduce the physical models and entanglement protocols underlying the framework.

\section{Entangled States, Noise Models and Entanglement Protocols}
\label{sec:ent_models_and_protocols}

In the following we present the entangled states, the noise models and the primitive entanglement protocols of swapping and purification.  

\subsection{Entangled States and Noise Models}
\label{sec:noise_models}
The routing framework developed in this work is based on entangled pairs that form links, elementary or virtual. The target state of an entangled pair is the maximally entangled Bell state
\begin{equation}
\ket{\beta_0}=\ket{\Phi^+}
=\frac{1}{\sqrt{2}}\left(\ket{00}+\ket{11}\right).
\label{EPR}
\end{equation}
Imperfections in entanglement generation, transmission, storage, and local operations generally transform the ideal pair into a mixed state $\rho$, characterized by the fidelity with respect to the target Bell state,
\begin{equation}
F(\rho)=\bra{\beta_0}\rho\ket{\beta_0}.
\end{equation}
For the Bell-diagonal noise models, considered in the following, the resulting states form one-parameter families fully specified by fidelity $F$. Fidelity, thus, provides a complete characterization of the state within each assumed noise model. Furthermore, swapping and purification (Section~\ref{sec:quantum_protocols}) preserve the state family, so the fidelity description covers both elementary and virtual links. We will, consequently, use fidelity to characterize the quality of an entangled pair throughout this work.  

We consider two Pauli noise models that are commonly used in quantum networks. Under symmetric Pauli noise, the three Pauli errors occur with equal probability, resulting in a Werner state. Under single-Pauli-error noise, only one error channel is retained, resulting in a two-Bell-state diagonal state. These two models are particularly useful here because they exhibit qualitatively different purification behavior while both admitting a single-parameter description in terms of the fidelity $F$. Their explicit density matrices and the corresponding channel descriptions are given in Appendix~\ref{app:quantum}. 

For the subsequent analysis, it is convenient to reparameterize the fidelity as
\begin{equation}
\xi_W(F)=\frac{4F-1}{3},
\qquad
\xi_{2B}(F)=2F-1,
\label{eq:polarization}
\end{equation}
for Werner and two-Bell-state diagonal states \footnote{The symmetric Pauli noise model leads to Werner states, and the single-Pauli-error noise model leads to two-Bell-state diagonal states. The terms are used interchangeably when we refer to noise models or states.}, respectively. We refer to $\xi$ as the \emph{Bell parameter} and use the same symbol whenever a statement applies to either state family. For both families the state is entangled if and only if $F>1/2$, i.e., $\xi>0$; at $F\le1/2$ the state is separable, and neither swapping nor purification can recover entanglement from it. This parametrization is especially useful because entanglement swapping composes multiplicatively in $\xi$, allowing end-to-end fidelity to be expressed linearly after a logarithmic transformation, as shown in Section~\ref{sec:quantum_protocols}.

Quantum-memory decoherence and imperfections of local operations can be incorporated into the description. Memory decoherence is included through the time-dependent elementary-pair fidelity introduced in Section~\ref{sec:operation_model}, while operation noise can be represented by a reduction of the Bell parameter by a factor $\eta\le1$ (Section~\ref{sec:quantum_protocols}). In the numerical results of Section~\ref{sec:results}, we assume ideal local operations and absorb source and operation imperfections into the baseline link fidelity $F_e^{(0)}$, similar to common entanglement-routing models~\cite{li2022_qpath_qleap,xiao2024_psc}.

\subsection{Entanglement Protocols}\label{sec:quantum_protocols}

\begin{figure}
\centering
\includegraphics[width=\linewidth, trim={0 0.8cm 0 .8cm}, clip]{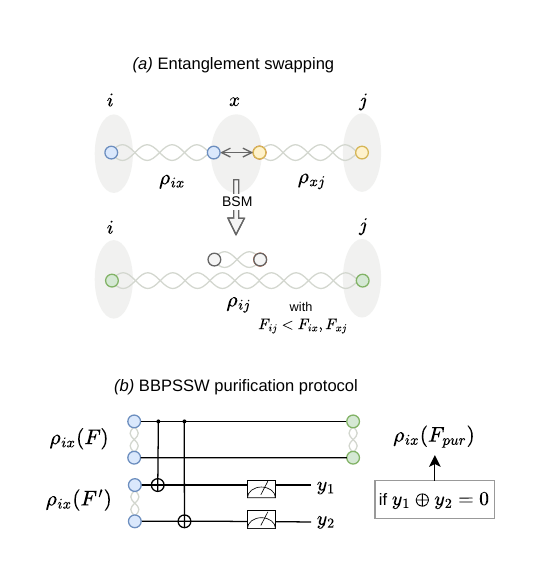}
\caption{ \textit{(a)} Entanglement swapping. A Bell-state measurement (BSM) at node $x$ transforms the adjacent entangled pairs $\rho_{ix}$ and $\rho_{xj}$ into the virtual pair $\rho_{ij}$. \textit{(b)} BBPSSW purification. Bilateral CNOT operations followed by measurement act on two entangled pairs, $\rho_{ix}(F)$ and $\rho_{ix}(F')$, spanning the same link, producing upon successful post-selection a purified pair $\rho_{ix}(\mathrm{F}_{\mathrm{pur}})$. Curly lines denote shared entanglement, while dots represent matter qubits stored in quantum memories.}\label{fig:2}
\end{figure}

Throughout this work, along a linear chain of $L$ elementary links $e=1,\ldots,L$, a (quantum) \emph{link} is the entanglement shared by two nodes and consists of one or more entangled pairs; it is \emph{elementary} if the nodes are adjacent and generate its pairs directly, \emph{virtual} if established by swapping. \emph{Raw pairs} are the pairs of an elementary link before purification. Swapping and purification act on pairs and thereby transform links. The optimization framework operates on two fundamental entanglement-processing primitives: entanglement swapping, which extends entanglement across non-neighboring nodes creating virtual links by consuming consecutive pairs, and entanglement purification, which probabilistically creates one pair of higher fidelity by consuming lower-fidelity entangled pairs. In this section, we summarize the action of both protocols (Fig.~\ref{fig:2}) for the two noise models considered in this work and derive the corresponding fidelity and success-probability update rules. These update rules constitute the physical layer on which we describe the purification schedules in Section~\ref{sec:schedules} upon which the optimization framework of Section~\ref{sec:algorithms} is built.

\subsubsection{Entanglement Swapping}

Entanglement swapping, originally proposed by \.{Z}ukowski \emph{et al.}~\cite{Zukowski1993} and experimentally demonstrated by Pan \emph{et al.}~\cite{Pan1998}, enables the establishment of long-distance entanglement by concatenating consecutive elementary links. 
Consider two consecutive elementary links $(i,x)$ and $(x,j)$ meeting at an intermediate node $x$, each carrying an entangled pair. A local Bell-stae measurment (BSM) is performed at node $x$ on the two qubits belonging to the respective pairs. This operation consumes the two constituent entangled pairs and establishes a new entangled pair directly between the distant endpoints $i$ and $j$. The resulting pair spans a \emph{virtual link} $(i,j)$, as illustrated in Fig.~\ref{fig:2}~(a).

For both Werner and two-Bell-state diagonal input states, the output remains within the same state family. Consequently, the action of the protocol is completely characterized by the corresponding fidelity update rule of the resulting (virtual) link. For Werner inputs,
\begin{equation}
F_{ij}=\mathrm{F}_{\mathrm{swap}}(F_{ix},F_{xj})=\frac{1-F_{ix}-F_{xj}+4F_{ix}F_{xj}}{3},
\label{swap_fidelity}
\end{equation}
while for two-Bell-state diagonal inputs,
\begin{equation}
F_{ij}=\mathrm{F}_{\mathrm{swap}}(F_{ix},F_{xj})=1-F_{ix}-F_{xj}+2F_{ix}F_{xj}.
\label{swap_fid}
\end{equation}
For equal input fidelities, $F_{ix}=F_{xj}=F$, a Werner output remains entangled only if
$F>(1+\sqrt{3})/4\approx0.683$. In contrast, for two-Bell-state diagonal states, the swapped pair remains entangled for every $F\geq1/2$; the single-Pauli-error noise model therefore exhibits no swapping threshold beyond the entanglement threshold.

Although Eqs.~(\ref{swap_fidelity}) and (\ref{swap_fid}) differ when expressed in terms of fidelity, they take an identical multiplicative form in the Bell parameter introduced in Section~\ref{sec:noise_models},
\begin{equation}
\xi_{ij}=\eta\,\xi_{ix}\,\xi_{xj},
\label{eq:swap_polarization}
\end{equation}
where $\eta=1$ for ideal local operations and $\xi$ denotes either $\xi_W$ or $\xi_{2B}$ according to the adopted noise model. This representation is central in the optimization framework developed in Section~\ref{sec:algorithms}, as it renders end-to-end fidelity composition multiplicative along the repeater chain.

\subsubsection{Purification}\label{sec:purification}

Entanglement purification probabilistically converts multiple low-fidelity entangled pairs into fewer pairs of higher fidelity --- a vital operation for preserving entanglement, and thus, a working link, under the aforementioned entanglement $F>1/2$ constraint. In this work, we assume the use of the BBPSSW protocol~\cite{Bennett1996}, which consumes two entangled pairs to produce one of improved fidelity. Given two pairs spanning the \emph{same} link---elementary or virtual---with fidelities $F$ and $F'$, bilateral CNOTs are applied and the target qubits are measured. If the measurement outcomes coincide, the remaining pair is retained with updated fidelity and success probability; otherwise both pairs are discarded.

For Werner states (symmetric Pauli errors), repeated application of BBPSSW requires an isotropic twirling step to restore the Werner form (Appendix~\ref{app:quantum}). For input fidelities $F,F'>1/2$, the output fidelity and success probability are
\begin{equation}
\begin{split}
&\mathrm{F}_{\mathrm{pur}}(F,F') = \frac{1-F-F'+10FF'}{5-2F-2F'+8FF'}, \\
&\mathrm{Pr}_{\mathrm{succ}}(F,F')
= \frac{5-2F-2F'+8FF'}{9}.
\end{split}
\label{W_pur}
\end{equation}

For two-Bell-state diagonal states (single-Pauli error), the output remains within the same state family, no twirling is required, and
\begin{eqnarray}
\mathrm{F}_{\mathrm{pur}}(F,F')&=&\frac{FF'}{1-F-F'+2FF'},\nonumber\\
\mathrm{Pr}_{\mathrm{succ}}(F,F')&=&1-F-F'+2FF'.
\label{2B_pur}
\end{eqnarray}

In both noise models, purification is probabilistic and symmetric in its two inputs, increasing fidelity at the expense of entangled-pair resources and success probability. These update rules quantify the fundamental trade-off between fidelity and throughput that underlies the optimization framework developed in Section~\ref{sec:algorithms}.

\section{Purification Schedule and Analytical Results}
\label{sec:schedules}

The BBPSSW purification protocol specifies how two entangled pairs are combined into a higher-fidelity pair, which is a primitive operation. Routing in quantum networks requires decisions on \emph{which} pairs to combined and in \emph{what order}. This choice is determined by the purification \emph{schedule} and constitutes one of the two fundamental design axes of entanglement-routing, alongside the order of purification and swapping operations. In this section, we introduce the schedule classes considered in this work and establish analytical results that characterize their capabilities. In particular, we derive feasibility bounds for pumping and tree schedules, and analyze the role of operation order under the throughput objective.

\subsection{Pumping and Tree Schedules}

Since each purification combines exactly two input pairs into one output pair, a purification schedule on a link can be represented as a rooted binary tree (Fig.~\ref{fig:xtra}). The leaves correspond to the raw entangled pairs initially available on a link, while each internal tree node represents one purification operation applied to the two pairs produced by its children.

We assume that the raw pairs of an elementary link $e$ share one baseline fidelity $F_e^{(0)}$ and that at most $c_e$ such pairs are available per routing cycle (Section~\ref{sec:operation_model}). A schedule is one tree with $b\le c_e$ leaves (raw pairs) and one root (a purified pair); its root fidelity and success probability depend on the tree's shape. A link runs the same tree $\lfloor c_e/b \rfloor$ times in parallel --- a forest of identical trees whose roots are the link's purified pairs, available for swapping --- so the link's width is $W_e=\lfloor c_e/b \rfloor$, and the $c_e \bmod b$ leftover raw pairs are unused.

The set of all binary trees defines the class of \emph{tree schedules}. Among the many possible schedules, two canonical constructions are important: pumping and balanced trees.

\begin{figure}
\centering
\includegraphics[width=\linewidth, trim={0.5cm 0.5cm 0.5cm 0.5cm}, clip]{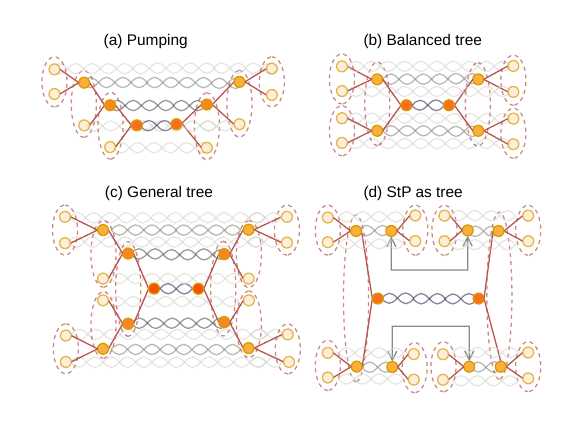}
\caption{ Examples of purification schedules represented as rooted binary trees. Leaves correspond to raw entangled pairs, while internal nodes represent BBPSSW purification operations. (a) Pumping repeatedly purifies one stored pair against freshly generated pairs, resulting in a maximally unbalanced tree. (b) A balanced tree recursively purifies pairs of equal purification depth. (c) A general tree schedule, of which pumping and balanced trees are particular instances. (d) An example of swapping before purifying (StP) having, under the single-Pauli-error model, similar effect as a balanced tree purification schedule --- initial pairs are purified, then swapped (gray arrows) and the final pairs are purified.}\label{fig:xtra}
\end{figure}

In \emph{pumping}---the maximally unbalanced tree---one stored pair is iteratively purified against a raw pair~\cite{dur1999_quantum_repeaters}. Writing $F_{\mathrm{pump},e}^{(k)}$ and $P_{\mathrm{pump},e}^{(k)}$ for the fidelity and cumulative success probability after $k$ pumping purification rounds, with $P_e^{(0)}=1$, the pumping recursion gives
\begin{eqnarray}
&F_{\mathrm{pump},e}^{(k+1)}=\mathrm{F}_{\mathrm{pur}}(F_{\mathrm{pump},e}^{(k)},F_e^{(0)}),\nonumber\\
&P_{\mathrm{pump},e}^{(k+1)}=P_{\mathrm{pump},e}^{(k)}\ 
\mathrm{Pr}_{\mathrm{succ}}(F_{\mathrm{pump},e}^{(k)},F_e^{(0)}),
\end{eqnarray}
where $\mathrm{F}_{\mathrm{pur}}$ and $\mathrm{Pr}_{\mathrm{succ}}$ are given by Eqs.~(\ref{W_pur}) or (\ref{2B_pur}). After $k$ rounds, pumping consumes $b=k+1$ raw entangled pairs, so a link pumps $\lfloor c_e/(k+1)\rfloor$ pairs to depth $k$ in parallel --- the pumping-in-parallel scheme of~\cite{chen2024_jsac}. Any literature work that selects a per-link purification \emph{depth} implicitly optimizes over the pumping class~\cite{li2022_qpath_qleap, zhang2025_link_config, ni2025_joint}.

In a \emph{balanced tree} (also referred to as \emph{self-purification}, or \emph{symmetric} purification in~\cite{chen2024_jsac}), two pairs that were produced by equally deep balanced trees (have equal fidelity $F_{\mathrm{bal},e}^{(\cdot)}$) are purified against one another. The recursion becomes
$F_{\mathrm{bal},e}^{(m+1)}=\mathrm{F}_{\mathrm{pur}}(F_{\mathrm{bal},e}^{(m)},F_{\mathrm{bal},e}^{(m)})$, $P_{\mathrm{bal},e}^{(m+1)}=(P_{\mathrm{bal},e}^{(m)})^2\ \mathrm{Pr}_{\mathrm{succ}}(F_{\mathrm{bal},e}^{(m)},F_{\mathrm{bal},e}^{(m)})$, with $F_{\mathrm{bal},e}^{(0)}=F_e^{(0)}$, while the raw-pair consumption doubles at every level, so that $m$ levels require $b=2^m$ raw pairs. Unlike pumping, where one input is always pinned at the baseline fidelity $F_e^{(0)}$, both inputs are themselves balanced trees recursively created with the same process. As Proposition~\ref{prop:fixed_points} shows, these two recursions converge to fundamentally different limits.

Throughout this paper we consider two \emph{schedule classes} over which optimization is performed: the \emph{pumping class}, in which schedules are completely specified by the purification depth $k$, and the \emph{tree class}, comprising the full family of binary purification trees. Pumping and balanced trees are simply two particular members of the latter. Consequently, optimization over the tree class searches over all feasible purification schedules rather than over a structured tree sub-class.

\subsection{Purification limits and feasibility bounds}

We define an entanglement-routing request as \emph{feasible} under a given method (schedule class and operation order) if the network can establish at least one end-to-end entangled pair whose fidelity $F_{\mathrm{e2e}}$ satisfies $F_{\mathrm{e2e}} \ge F_{\mathrm{th}}$, using the available resources (raw entangled pairs). The required fidelity threshold $F_{\mathrm{th}}$ depends on the application. 

We now investigate how the underlying noise model influences the effectiveness of different purification schedules. Under the single-Pauli-error noise model, the purification map is associative in a transformed variable (Appendix~\ref{app:dp_optimality}; see also~\cite{chen2024_jsac}), so every tree with the same number of leaves produces the same output state (Lemma~1 of~\cite{chen2024_jsac}) and, under the per-copy success accounting of Section~\ref{sec:algorithms}, the same cumulative success probability (Appendix~\ref{app:dp_optimality}). Consequently, pumping, balanced trees, and arbitrary tree schedules are equivalent in fidelity, success probability, and raw-pair cost: the number of raw pairs is the only parameter that matters. Under symmetric Pauli noise, however, associativity is lost and the shape of the purification tree directly affects the achievable fidelity. For example, starting with baseline fidelity $F_e^{(0)}=0.8$, pumping four raw pairs achieves $F=0.864$, but purifying according to a balanced tree (Fig.~\ref{fig:xtra}) achieves $F=0.874$; the corresponding fidelities resulting from eight raw pairs are $0.870$ and $0.905$. These observations suggest that pumping converges to a finite limit, whereas balanced trees continue to improve. Proposition~\ref{prop:fixed_points} establishes this formally.

\begin{proposition}[Purification limits]\label{prop:fixed_points}
Let $F_e^{(0)}>1/2$ and $\eta=1$. \textit{(i)} Under the single-Pauli-error model the pumping fidelity $F_{\mathrm{pump},e}^{(k)}$ increases monotonically in $k$ and converges to $1$. \textit{(ii)} Under the symmetric Pauli noise model it increases monotonically to a limit $F_{\mathrm{pump}}^{\ast}(F_e^{(0)})<1$, the unique root in $(F_e^{(0)},1)$ of a quadratic given in Appendix~\ref{app:convexity}. \textit{(iii)} Under both models the balanced-tree fidelity $F_{\mathrm{bal},e}^{(m)}$ increases monotonically in $m$ and converges to $1$.
\end{proposition}

Proposition~\ref{prop:fixed_points} has a direct consequence for the feasibility of pumping-based solutions. Under the symmetric Pauli noise model, which emerges as a natural benchmarking context by not presupposing a dominant error channel, every pumped link satisfies $F_{\mathrm{pump},e}^{(k)}<F_{\mathrm{pump}}^{\ast}(F_e^{(0)})$ at any depth $k$ (Proposition~\ref{prop:fixed_points}(ii)); since swapping composes Bell parameters multiplicatively by Eq.~\eqref{eq:swap_polarization} and is lossless only at $\xi=1$, the per-link limits compose into a bound on every PtS sequence of the pumping class:
\begin{equation}
F_{\mathrm{e2e}} \;\le\; \frac{1}{4} + \frac{3}{4}\prod_{e=1}^{L} \xi_{W}\!\left(F_{\mathrm{pump}}^{\ast}\!\left(F_e^{(0)}\right)\right).
\label{eq:pts_ceiling}
\end{equation}
For $F_e^{(0)}=0.8$ we obtain $F_{\mathrm{pump}}^{\ast}=0.870$, and assuming the same $F_e^{(0)}=0.8$ over $L=4$ links, a ceiling of $0.601$; for $F_e^{(0)}=0.9$ we correspondingly obtain $0.943$ per link and $0.798$ end-to-end (e2e). A fidelity threshold above the bound of Eq.~\eqref{eq:pts_ceiling} renders every PtS solution of the pumping class infeasible, irrespective of the depths selected. The bound is therefore determined entirely by the raw elementary pair fidelities $F_e^{(0)}$ and is independent of available link capacities. Section~\ref{sec:results} shows that it predicts, exactly, the measured blocking of every PtS order for the pumping class.

\subsection{Escaping the Pumping Bound}

The pumping bound of the symmetric Pauli noise model can be escaped in two ways, each relaxing one of the two design restrictions in the two dimensions considered in this work: the restriction to pumping purification schedule and the restriction to the PtS operation order. These two mechanisms can be combined.

\subsubsection{Tree Schedules}
The first mechanism enlarges the purification schedule class by allowing general tree schedules rather than pumping. By Proposition~\ref{prop:fixed_points}\textit{(iii)}, tree schedules can asymptotically approach unit fidelity when raw pairs are available. Their achievable fidelity is therefore limited only by the available link capacity rather than by the raw pair fidelity. 
Let $\Phi_e(c)$ denote the highest fidelity attainable on link $e$ using at most $c_e$ raw pairs.  Since $\mathrm{F}_{\mathrm{pur}}$ is increasing in both operands (Appendix~\ref{app:dp_optimality}), $\Phi_e(c_e)$ obeys the recursion $\Phi_e(1)=F_e^{(0)}$, $\Phi_e(c_e)=\max_{1\le a<c_e} \mathrm{F}_{\mathrm{pur}}\bigl(\Phi_e(a),\Phi_e(c_e-a)\bigr)$, and the resulting per-link maxima compose into a capacity-\emph{dependent} end-to-end bound on every PtS sequence of the tree class:
\begin{equation}
F_{\mathrm{e2e}} \;\le\; \frac{1}{4} + \frac{3}{4}\prod_{e=1}^{L} \xi_{W}\!\bigl(\Phi_e(c_e)\bigr).
\label{eq:tree_ceiling}
\end{equation}
Unlike Eq.~\eqref{eq:pts_ceiling}, this bound recedes as capacity grows; Section~\ref{sec:results} verifies both.
Note that both Eq.~\eqref{eq:pts_ceiling} and Eq.~\eqref{eq:tree_ceiling} assume the PtS order. In contrast, post-swap purification acts on the composed end-to-end pair, which, by Proposition~\ref{prop:fixed_points}\textit{(iii)}, can itself be purified beyond any composition of per-link maxima. Section~\ref{sec:results} quantifies how rarely such interleaved sequences actually exceed these bounds.

\subsubsection{Post-swap purification}

Whereas tree schedules maximize the fidelity attainable on individual elementary links prior to swapping, post-swap purification operates on \emph{virtual links} established across distant nodes. How to numerically optimize this was studied in~\cite{vecino2026_strategy}.

Consider an arbitrary virtual link $(i,j)$ formed by swapping operations along a sub-path of width $W$, whose pairs have fidelity $F$ and cumulative success probability $P$ (with $P=1$ if no probabilistic purification preceded). Note we suppress subscript ($ij$) for simplicity. The expected throughput over this virtual link is given by $W \cdot P$. Then assume a round of $\lfloor W/2\rfloor$ post-swap purifications execute in parallel, each consuming two pairs and producing one of higher fidelity, or none on failure. The resulting virtual link has fidelity $\mathrm{F}_{\mathrm{pur}}(F,F)$, available width $\lfloor W/2\rfloor$, and updated cumulative success probability $P^2\,\mathrm{Pr}_{\mathrm{succ}}(F,F)$. The expected throughput of the post-swap purified link is therefore scaled by a factor of
\begin{equation}
\frac{\lfloor W/2\rfloor \cdot P^2\,\mathrm{Pr}_{\mathrm{succ}}(F,F)}{W \cdot P} \;\le\; \frac{1}{2}\,P\,\mathrm{Pr}_{\mathrm{succ}}(F,F) \;<\; 1.
\end{equation}
Because $P \le 1$ and $\mathrm{Pr}_{\mathrm{succ}}(F,F) < 1$ for any state, this scaling factor is strictly less than unity. Consequently, post-swap purification strictly reduces the expected throughput. It can, however, increase the fidelity and carry an otherwise infeasible configuration across the required fidelity threshold $F_{\mathrm{th}}$.

Operation order and purification schedule are not entirely orthogonal: even when individual links are restricted to pumping, interleaving swapping with post-swap purification synthesizes a distributed, multi-level purification tree across the network. This explains intuitively how post-swap purification escapes the pumping bound. For instance, swapping two pumped pairs with pairs of the next link and subsequently purifying the resulting virtual pairs against one another executes a balanced purification on these virtual pairs (Fig.~\ref{fig:xtra}(d)). By accessing this recursive structure, interleaved operation sequences can drive fidelity toward unity ($F \to 1$), making feasible instances that would otherwise be blocked under the strict PtS pumping bound of Eq.~\eqref{eq:pts_ceiling}. Section~\ref{sec:results} quantifies how often this mechanism is invoked in practice.

\subsection{Implications for the Optimization Framework}

Under the single-Pauli-error noise model, neither escape mechanism is required to overcome a pumping limit, since Proposition~\ref{prop:fixed_points}\textit{(i)} shows that pumping itself converges to unit fidelity. Consequently, pumping and tree schedules are equivalent in terms of feasibility, and any differences between them arise only in throughput. Under symmetric Pauli noise, however, the two escape mechanisms play different roles: tree schedules expand the feasible region and can also improve throughput on jointly served instances by saving raw-pair consumption; in contrast, post-swap purification expands feasibility alone, as it inevitably reduces throughput to increase fidelity above otherwise unreachable targets.

These analytical observations motivate the optimization framework developed in the next section. The proposed formulations can be instantiated for both noise models and for both purification schedule classes, MILP covers only PtS operation order while DP further supports arbitrary orders. Section~\ref{sec:results} systematically evaluates all combinations, allowing the individual impact of the noise model, purification schedule, and operation order to be assessed independently.
 


\section{Network Operation Model}
\label{sec:operation_model}

The network operates in discrete routing cycles. Within each cycle, operations proceed in two phases. In the \emph{generation phase}, of duration $\tau_{\mathrm{gen}}$, adjacent physical nodes (e.g. connected with fibers), which operate as repeaters, establish raw entangled pairs across each elementary link; the hardware provisions pairs at a baseline generation rate $R_{\mathrm{gen},e}$, so that elementary link $e$ accumulates up to $c_e = \lfloor R_{\mathrm{gen},e} \cdot \tau_{\mathrm{gen}} \rfloor$ raw pairs, assuming that there are available buffers. In the \emph{operation phase}, of duration $\tau_{\mathrm{op}}$, nodes apply purification and swapping to the accumulated pair pool to deliver end-to-end pairs. The cycle length is $\tau_{\mathrm{slot}} = \tau_{\mathrm{gen}} + \tau_{\mathrm{op}}$. This phased, "stop-and-wait" operation model (Fig.~\ref{fig:3}) is common practice in entanglement routing~\cite{li2022_qpath_qleap, wang2026_sftrap, koutsopoulos2024_milcom}.

\begin{figure}
\centering
    \includegraphics[ width=\linewidth, trim={1.5cm 0 0cm 0cm}, clip]
    {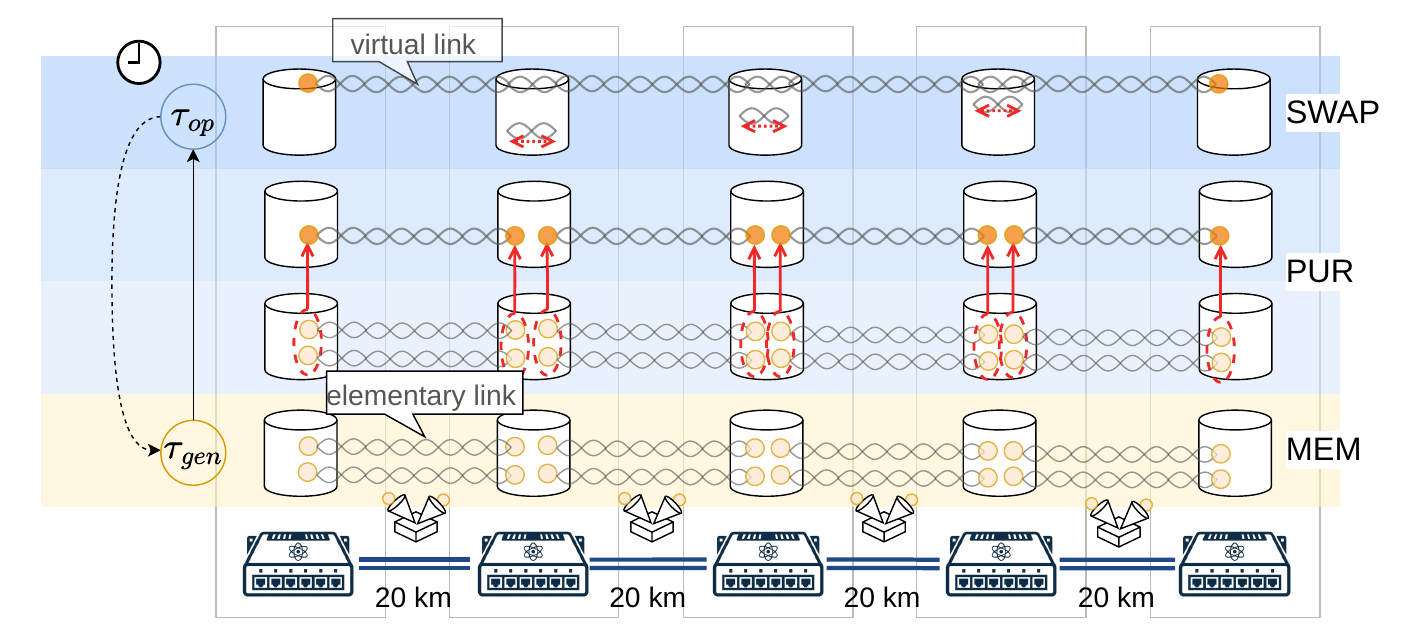}
    \caption{Network operation model. Entangled pairs generated at rate $R_{\mathrm{gen},e}$ during the generation phase of duration $\tau_{\mathrm{gen}}$ are stored in memory and later used for purification and entanglement swapping in the operation phase of duration $\tau_{\mathrm{op}}$, yielding an effective link capacity $c_e=\lfloor R_{\mathrm{gen},e}\cdot \tau_{\mathrm{gen}} \rfloor$. Then, the cycle of duration $\tau_{\mathrm{slot}}$ repeats.}
    \label{fig:3}
\end{figure}

The operation phase is dominated by classical communication: every purification and swap requires exchanging measurement outcomes between link endpoints at medium propagation speed, for fiber $v_f \approx 2\times 10^8$~m/s. Swapping and post-swap purification propagate outcomes across virtual links up to $L_{\mathrm{total}}/v_f$. Because each post-swap purification round halves available width, at most $k_{\mathrm{post}} \le \lfloor \log_2 \min_e c_e \rfloor$ such rounds occur along any root-to-leaf path. For maximum per-link depth $k_{\max}$, the operation phase is bounded (if operations are done sequentially) by
\begin{equation}
\tau_{\mathrm{op}}\;\lesssim\; k_{\max}\frac{L_0}{v_f}\;+\;\left(1+k_{\mathrm{post}}\right)\frac{L_{\mathrm{total}}}{v_f}.
\label{eq:top}
\end{equation}
In our simulations, pumping-based methods cap the per-link depth at $k_{\max}=20$, whereas tree schedules are bounded by the link capacity alone ($b \le c_e$). Neither limit is active in the reported results: the deepest selected pumping schedule uses $10$ purification rounds and the deepest tree schedule has sequential depth $9$. Over the parameter ranges of Section~\ref{sec:results} ($L \le 6$, $c_e \le 150$, $L_0 = 20$~km), Eq.~\eqref{eq:top} gives $\tau_{\mathrm{op}} \le 6.8$~ms, against the $\tau_{\mathrm{gen}} = 10$~ms we set.

For fidelity, both durations are evaluated against the memory coherence time $\tau_{\mathrm{coh}}$. We model decoherence through a worst-case approximation where pairs remain in memory for the full generation phase, yielding a baseline fidelity on link $e$ of
\begin{equation}\label{eq:baseline-fidelity}
F_e^{(0)} = \mu_e ^{\left(-\tau_{\mathrm{gen}} /{\tau_{\mathrm{coh}}}\right)} ,
\end{equation}
where $\mu_e$ is the link source fidelity and $\tau_{\mathrm{coh}}$ is the memory coherence time. 
Decoherence accrued during the operation phase is bounded by $\exp(-\tau_{\mathrm{op}}/\tau_{\mathrm{coh}})$ and neglected. In the simulations, we fix $\tau_{\mathrm{gen}}=10\text{ ms}$ and $\tau_{\mathrm{coh}}=1.46\text{ s}$, a value within the range reported for contemporary quantum memories~\cite{liu2026longlived, liu2024memorymemory, Ortu2022, LagoRivera2021}. The studied capacities $c_e \in [10, 150]$ correspond to generation rates between $1$ and $15\text{ kpairs/s}$, realistic for current and emerging generators. Thus, the modeled fidelity discount is $1-e^{-\tau_{\mathrm{gen}}/\tau_{\mathrm{coh}}} \approx 0.7\%$, and the neglected operation-phase discount is at most $0.5\%$. Both are small because $\tau_{\mathrm{gen}}, \tau_{\mathrm{op}} \ll \tau_{\mathrm{coh}}$, satisfying the operational regime assumed by the model. The algorithms of Section~\ref{sec:algorithms} operate directly on $c_e$ and $F_e^{(0)}$.

The duration of the operation phase play a role in the duty cycle $\tau_{\mathrm{gen}}/\tau_{\mathrm{slot}}$, and hence directly affects throughput (ratio of delivered end-to-end pairs over time). A pipelined or asynchronous architecture~\cite{yang2024_async, haldar2025_quasilocal}, in which generation continues during the operation phase, would raise that rate. It would also invalidate the notion of a fixed pool, and with it, the width variable on which the formulations of Section~\ref{sec:algorithms} rest. We thus retain the stop-and-wait model here and plan to study the pipelined regime in future work.


\section{Problem Formulation and Algorithms}
\label{sec:algorithms}

We formulate the capacity-aware routing problem on a repeater chain, and then present a framework to solve it.

\subsection{Problem Formulation}\label{sec:formulation}

We model the quantum repeater chain as a linear network with node set $V=\{1,\ldots,N\}$ and elementary link set $E=\{(i,i+1)\mid 1\leq i<N\}$, with $L=|E|$. Nodes $1$ and $N$ denote the end-to-end source and destination. For algebraic convenience, we index elementary links interchangeably by their endpoints $(i,i+1)$ or by index $e\in E$. As established in Section~\ref{sec:operation_model}, each elementary link $e$ is characterized by a capacity $c_e=\lfloor R_{\mathrm{gen},e}\ \tau_{\mathrm{gen}}\rfloor$ of generated (raw) pairs and a baseline fidelity $F_e^{(0)}$, and these constitute the inputs to the controller running the algorithms discussed here.

The controller establishes end-to-end entanglement by executing the two primitive operations, each consuming two input pairs to produce a single output pair:

\begin{itemize}
    \item \textbf{Entanglement Swapping ($\textsc{SWAP}$):} Consumes two entangled pairs spanning $(i,x)$ and $(x,j)$ (which may be elementary or virtual) to establish a new entangled pair across the virtual link $(i,j)$.
    \item \textbf{Entanglement Purification ($\textsc{PUR}$):} Consumes two entangled pairs spanning the same endpoints $(i,j)$ to yield a single pair of higher fidelity. When $j = i+1$, purification operates on pairs of an elementary link, raw or previously purified (according to the purification schedule); when $j > i+1$, it operates on pairs established across a virtual link (post-swap purification).
\end{itemize}
Every operand is either a raw pair on an elementary link or the output of a prior operation. A routing solution is therefore formalized as a \emph{routing tree} whose leaves represent raw entangled pairs and whose internal nodes represent $\textsc{SWAP}$ or $\textsc{PUR}$ operations. For example, a balanced tree corresponds to $\textsc{PUR}(\textsc{PUR}(r_1,r_2),\textsc{PUR}(r_3,r_4))$. A routing tree builds one end-to-end pair; as with the link-level forests of Section~\ref{sec:schedules}, it is run in parallel as many times as the raw pairs of its links allow, and post-swap purifications consume copies of the same subtree. The number of parallel copies is the end-to-end width $W$. When referring to operations at the link level, we adopt the endpoint shorthand $\textsc{SWAP}(i,x,j)$ and $\textsc{PUR}(i,j)$. Any linear execution sequence $s=[o_1,\ldots,o_M]$ of these operations that respects the tree's data dependencies constitutes a \emph{routing sequence}.

Executing $s$ delivers $W(s)$ end-to-end pairs (width) of fidelity $F_{\mathrm{e2e}}(s)$ with cumulative success probability $P(s)$ per pair. The routing objective is to maximize the expected end-to-end throughput
\begin{equation}
\max_{s}\; T_{\mathrm{e2e}}(s)= W(s)\cdot P(s)
\label{eq:objective}
\end{equation}
subject to the fidelity constraint $F_{\mathrm{e2e}}(s)\geq F_{\mathrm{th}}$.

Maximizing \eqref{eq:objective} over valid routing sequences $s$ requires navigating the two design axes formalized in Section~\ref{sec:schedules}. \emph{Operation order}: Under purify-then-swap (PtS), all elementary operations $\textsc{PUR}(i,i+1)$ strictly precede any $\textsc{SWAP}$, and no post-swap purification $\textsc{PUR}(i,j)$ ($j > i+1$) occurs. Relaxing PtS admits arbitrary interleavings, enabling post-swap purification across virtual links. \emph{Purification schedule}: For each elementary link $e \in E$, the schedule maps capacity $c_e$ to an operational triple $(F_e, W_e, P_e)$ denoting post-purification fidelity, available width, and cumulative success probability. Under the pumping class, the routing decision reduces to a depth vector $\mathbf{k} = [k_1,\ldots,k_L]$, with $b_e = k_e + 1$ raw pairs per purified pair and $W_e = \lfloor c_e / b_e \rfloor$. Under the tree class, each link selects one entry of its choice set $K_e$ (Section~\ref{sec:milp}). Our formulations apply the same operations to every pair of a link, so every link, elementary or virtual, is described by a single $(F_e, W_e, P_e)$ triple; solutions that treat pairs of one link differently lie outside the framework (Appendix~\ref{app:dp_optimality}).

This structural distinction organizes the algorithmic framework that follows. Fixing the operation order to PtS, we solve \eqref{eq:objective} exactly for both schedule classes using a single Mixed-Integer Linear Program (MILP) instantiated with the choice set corresponding to each class (Section~\ref{sec:milp}). For online operation we also provide a fast greedy heuristic (Section~\ref{sec:heuristic}). Admitting arbitrary operation orders expands the search space to multi-hop virtual-link purification, which we optimize exactly using Dynamic Programming (DP) for either schedule class (Section~\ref{sec:dp}).

Under the PtS restriction, the problem exhibits a Multiple-Choice Knapsack structure, which is NP-hard in general~\cite{kellerer2004_knapsack}: each link contributes exactly one schedule from its discrete choice set $K_e$; each choice consumes integer capacity; and each operation contributes additively in the log-Bell domain toward the fidelity threshold. This motivates exact methods that exploit problem structure rather than general-purpose integer solvers.

The formulations below are noise model-agnostic: the noise model enters only through the pre-computed coefficients, and any model and objective satisfying the structural properties of Appendix~\ref{app:convexity} are admissible, including both models of Section~\ref{sec:noise_models}. Collectively, the choice-set abstraction and the methods built on it --- the two MILP instantiations, the greedy heuristic, and the two order-exact DP --- constitute our routing \emph{framework}.

\subsection{Routing Framework}
\label{sec:routing_algorithms}

\subsubsection{Mixed-Integer Linear Programming Formulation}
\label{sec:milp}

We provide an optimal algorithm for the capacity-aware routing problem under the PtS order, based on an MILP formulation. Under PtS, both the fidelity composition of swapping and the cumulative success probability of purification are multiplicative; taking logarithms makes both linear and thus admissible in a MILP.

Under PtS every link completes its purification before any swap, so the per-link purification decision separates from the chain-level selection: we pre-compute, for each link $e$, a discrete \emph{choice set} $K_e$ of candidate schedules, each described by its root fidelity  $F_{e,k}$, forest width $W_{e,k}$, and cumulative success probability $P_{e,k}$, with the derived Bell parameter $\Xi_{e,k}=\xi(F_{e,k})$ of Eq.~\eqref{eq:polarization}. Throughout the formulations, uppercase symbols denote inputs and physical quantities, lowercase symbols decision variables and indices. In the pumping class, $k$ is the purification depth, and $F_{e,k}=F_{\mathrm{pump},e}^{(k)}$, $W_{e,k}=\lfloor c_e/(k+1)\rfloor$, and $P_{e,k}=P_{\mathrm{pump},e}^{(k)}$, ranging over depths that do not exhaust the capacity ($k\le c_e$). In the tree class, $k$ indexes the link's \emph{purification-tree frontier}. For a tree with $b$ leaves the width is fixed (number of trees in the forest), $W_{e,k}=\lfloor c_e/b \rfloor$, while the root fidelity and success probability follow from applying the update rules of Section~\ref{sec:quantum_protocols} at each internal tree node. A triple $(F,W,P)$ \emph{dominates} another if it is at least as large in every component of the triple and strictly larger in one; a dominated schedule can never be preferred, and the frontier retains the non-dominated triples over all trees with $b\le c_e$ (Appendix~\ref{app:dp_optimality} formalizes this). We construct the frontier once per link with the schedule DP of Chen and Jia~\cite{chen2024_jsac}, which enumerates tree shapes on a discretized fidelity grid and guarantees that, for each leaf count, the best achievable root fidelity is missed by at most the grid resolution, $10^{-3}$ in our simulations; the pumping choice set is embedded in the frontier by construction (Section~\ref{sec:schedules}), and in the instances of Section~\ref{sec:results} the frontiers contain of the order of $10^2$ points per link, so $|K_e|$, the number of schedule choices per link, stayed small. The threshold $F_{\mathrm{th}}$ maps to $\Xi_{\mathrm{th}}=\xi(F_{\mathrm{th}})$.

Along the chain the success probabilities $P_{e,k}$ and Bell parameters $\Xi_{e,k}$ of the selected schedules compose multiplicatively, while the widths compose through a minimum, $W=\min_e W_{e,k}$ being the number of end-to-end pairs. The logarithm linearizes the products; the minimum is handled by an auxiliary continuous variable $\omega$ that equals $\log W$ at the optimum. The only decision variables are the binaries $x_{e,k}$, equal to $1$ if link $e$ adopts schedule $k$ in the $K_e$ set. The MILP is formulated as follows:
\begin{equation}
\max \; \omega + \sum_{e \in E} \sum_{k \in K_e} x_{e,k} \, \log P_{e,k}
\label{eq:milp_objective}
\end{equation}
subject to
\begin{align}
C1: \; & \sum_{k \in K_e} x_{e,k} = 1, && \forall e \in E, \label{eq:milp_c1} \\
C2: \; & \omega \leq \sum_{k \in K_e} x_{e,k} \, \log W_{e,k}, && \forall e \in E, \label{eq:milp_c2} \\
C3: \; & \sum_{e \in E} \sum_{k \in K_e} x_{e,k} \, \log \Xi_{e,k} \;\geq\; \log \Xi_{\mathrm{th}}. \label{eq:milp_c3}
\end{align}
The objective is the logarithm of Eq.~\eqref{eq:objective} under PtS, $\log T_{\mathrm{e2e}}=\log W+\sum_e \log P_{e,k(e)}$, where $k(e)$ is the schedule selected on link $e$, since the chain's success probability is the product of the selected schedules' probabilities. $C1$ selects exactly one schedule per link. $C2$ bounds $\omega$ by the minimum provisioned width across the chain, which the objective then drives to equality. $C3$ enforces the end-to-end fidelity threshold as a linear constraint on the log-Bell parameters, by Property~2 of Appendix~\ref{app:convexity}. The formulation is identical for the two purification schedule classes; only the candidate schedules and their coefficient triples $(W_{e,k},P_{e,k},\Xi_{e,k})$ change. We refer to the two instantiations as the pumping-based and tree-based MILP.

The underlying resource allocation is NP-hard in general, but the MILP is tractable at the scales of interest. Pre-computing the logarithmic coefficients strips the physical layer out of the constraint matrix and decouples the solver from both the noise model and the schedule class. In practice, the linear relaxation is tight and both instantiations solve in milliseconds for the chain lengths we evaluate (Section~\ref{sec:results}).

\subsubsection{Throughput-First Greedy Heuristic}
\label{sec:heuristic}

Exact formulations like the MILP defined above may be unusable to a controller operating under strict timing constraints or serving many concurrent demands. Thus, we also give a greedy heuristic for the PtS operation order. Baselines such as Q-PATH and Q-LEAP~\cite{li2022_qpath_qleap} allocate purification against uniform per-link fidelity targets, which blinds the controller to the capacity bottleneck and exhausts entangled pairs on heterogeneous links. Instead, we select in each step the link maximizing the capacity-aware marginal utility
\begin{equation}
u_e = \frac{\Delta \log \xi_e}{\left|\Delta T_{\mathrm{e2e}}\right|},
\label{eq:utility}
\end{equation}
the ratio of the marginal gain $\Delta \log \xi_e$ in log-Bell parameter to the magnitude $\left|\Delta T_{\mathrm{e2e}}\right|$ of the expected-throughput reduction that the purification step causes. Placing the throughput penalty in the denominator suppresses purification choices that prematurely reduce the global bottleneck width. Starting from the raw configuration, the heuristic advances, on the link of highest $u_e$, one step along that link's choice set --- the next pumping depth, or the next frontier point in decreasing width --- until $F_{\mathrm{e2e}}\ge F_{\mathrm{th}}$, and repeats the construction for each candidate bottleneck width from $\min_e c_e$ down to $1$, returning the feasible configuration of highest throughput. The outer sweep costs a factor $c_{\mathrm{max}}$ and repairs the greedy choices that a single pass makes at a fixed width (Appendix~\ref{app:counterexample}). The overall complexity is $O(L \cdot c_{\mathrm{max}} \cdot |K_{\mathrm{max}}|)$, with $|K_{\mathrm{max}}|$ the largest per-link choice set; the complexity formula holds for the two classes, with $|K_{\mathrm{max}}|$ taking a different value.


\subsubsection{Dynamic Programming over Arbitrary Operation Orders}
\label{sec:dp}

The MILP is exact under PtS, but a fixed operation order misses valid solutions: a virtual link established by swapping may itself be purified, and deferring purification until after a swap can carry a configuration across the threshold that no PtS allocation reaches (Section~\ref{sec:schedules}). Under PtS each link's schedule tree is flattened into one entry of its choice set, and a solution is a vector of per-link selections--- the variables $x_{e,k}$ the MILP searches. With the operation order free, a solution is a full routing tree of Section~\ref{sec:formulation}, whose internal nodes now include swaps, elementary link and virtual link (post-swap) purifications. We search this space exactly with a DP over contiguous sub-paths.

The \emph{DP state} of the sub-path from node $i$ to node $j$ is a set $\mathcal{C}_{i,j}$ of \emph{configurations} $\sigma=(F,W,P,s)$: the triple $(F,W,P)$ of Section~\ref{sec:milp} --- now the fidelity, width, and cumulative success probability of the virtual link $(i,j)$--- extended by the routing tree $s$ that builds one pair of that link from the raw pairs of elementary links $i,\ldots,j-1$, the unit replicated across the link's width. The tree determines the triple; the DP stores the triple alongside it to avoid re-evaluation, and dominance is tested on the triple only. The DP works with fidelity rather than the Bell parameter: it applies the exact composition maps $\mathrm{F}_{\mathrm{swap}}$ and $\mathrm{F}_{\mathrm{pur}}$ of Section~\ref{sec:quantum_protocols} directly and tests $F\ge F_{\mathrm{th}}$
as is, with no linearization to preserve. No single scalar summarizes a configuration --- a lower fidelity may be preferable if it retains width or probability --- so each $\mathcal{C}_{i,j}$ is a set of non-dominated configurations under the dominance relation, applied to sub-path configurations instead of link schedules (as done in Section~\ref{sec:milp}), rather than a single optimum configuration.

The base pool of each elementary link is its choice set, $\mathcal{C}_{i,i+1}=\{(F_{e,k},W_{e,k},P_{e,k},s_{e,k}):k\in K_e\}$ with $s_{e,k}$ the corresponding routing tree, so the DP and the MILP of a given schedule class evaluate identical link-level options and differ only in the operation order; the frontier filtering of Section~\ref{sec:milp} is the DP's own filter at $l=1$. The recursion merges adjacent sub-paths. For each span $(i,j)$ and each split node $x$ in between, every pair $(\sigma_A,\sigma_B)\in\mathcal{C}_{i,x}\times\mathcal{C}_{x,j}$ is swapped, producing $\bigl(\mathrm{F}_{\mathrm{swap}}(F_A,F_B),\,\min(W_A,W_B),\,P_AP_B,\,\textsc{SWAP}(s_A,s_B)\bigr)$; the result is then post-swap purified while $W\ge2$, each round updating $\sigma$ and producing $\bigl(\mathrm{F}_{\mathrm{pur}}(F,F),\,\lfloor W/2\rfloor,\,P^{2}\,\mathrm{Pr}_{\mathrm{succ}}(F,F),\,\textsc{PUR}(s,s)\bigr)$, and every intermediate configuration is retained, including those below the threshold: the optimal routing may hold an inner span below $F_{\mathrm{th}}$ on purpose, because a sibling span needs its width and probability rather than its fidelity margin, so threshold-based truncation of the purification loop can miss the optimum. Because the operands of a swap occupy disjoint elementary links, no configuration can exceed any link's capacity, and feasibility is not tracked explicitly.

To control the growth of the pool, the dominance filter is applied at each span: a configuration is discarded if another is at least as large in all of $F,W$ and $P$ and strictly larger in one. Dominance is tested on the three components and not on the scalar $T=W P$: widths compose through a minimum and probabilities through a product, so two configurations of equal $T$ need not be equally valuable downstream. Each operation is monotone in $F,W$ and $P$ separately (Appendix~\ref{app:dp_optimality}), which makes the filter sufficiency-preserving --- discarding a dominated configuration never discards the optimum --- and the same monotonicity is what justifies filtering the per-link frontier in Section~\ref{sec:milp}. Once the full chain is evaluated, the DP selects the configuration of the (source,destination) $\mathcal{C}_{1,N}$ set with maximum $T$ among those satisfying $F \geq F_{\mathrm{th}}$, together with its optimal routing tree $s^*$.

\begin{algorithm}
\caption{Dynamic Programming over Arbitrary Operation Orders}
\label{alg:dp}
\begin{algorithmic}[1]
\REQUIRE Linear chain of $L$ links; per-link choice sets $K_e$; threshold $F_{\mathrm{th}}$
\ENSURE Optimal routing tree $s^*$ and throughput $T_{\mathrm{e2e}}$
\STATE $\mathcal{C}_{i,i+1} \leftarrow \{(F_{e,k},W_{e,k},P_{e,k},s_{e,k}) : k\in K_e\}$, $\forall e=(i,i+1)\in E$
\FOR{sub-path length $\ell = 2$ \TO $L$}
   \FOR{start node $i = 1$ \TO $L-\ell+1$, \quad $j \leftarrow i+\ell$}
      \STATE $\mathcal{C} \leftarrow \emptyset$
      \FOR{split node $x = i+1$ \TO $j-1$}
        \FOR{each $(\sigma_A,\sigma_B) \in \mathcal{C}_{i,x} \times \mathcal{C}_{x,j}$}
            \STATE $\sigma \leftarrow \bigl(\mathrm{F}_{\mathrm{swap}}(F_A,F_B),\ \min(W_A,W_B),\ P_AP_B,$ \\
            \hspace*{1.2em} $\textsc{SWAP}(s_A,s_B)\bigr)$; \; $\mathcal{C} \leftarrow \mathcal{C}\cup\{\sigma\}$
            \WHILE{$W(\sigma) \geq 2$}
                \STATE $(F, W, P, s) \leftarrow \sigma$
                \STATE $\sigma \leftarrow \bigl(\mathrm{F}_{\mathrm{pur}}(F,F),\ \lfloor W/2\rfloor,\ P^{2}\,\mathrm{Pr}_{\mathrm{succ}}(F,F),$ \\
                \hspace*{1.2em} $\textsc{PUR}(s,s)\bigr)$; \; $\mathcal{C} \leftarrow \mathcal{C}\cup\{\sigma\}$
            \ENDWHILE
         \ENDFOR
      \ENDFOR
      \STATE $\mathcal{C}_{i,j} \leftarrow \textsc{ParetoFilter}_{(F,W,P)}(\mathcal{C})$
   \ENDFOR
\ENDFOR
\STATE $\sigma^* \leftarrow \arg\max\{\,W\cdot P : \sigma\in\mathcal{C}_{1,N},\ F\ge F_{\mathrm{th}}\,\}$
\RETURN $s^*=s(\sigma^*)$ and $T_{\mathrm{e2e}}=W(\sigma^*)\cdot P(\sigma^*)$
\end{algorithmic}
\end{algorithm}

\begin{theorem}\label{thm:dp}
Instantiated with a per-link choice set, Algorithm~\ref{alg:dp} returns a routing tree of maximum expected throughput among all routing trees that draw each link's schedule from that set and satisfy the fidelity constraint.
\end{theorem}

The proof, in Appendix~\ref{app:dp_optimality}, establishes optimal substructure over split points and shows that the Pareto filter is sufficiency-preserving. The DP evaluates $O(L^2)$ spans, each merging $O(L)$ splits over pairs of retained configurations and appending $O(\log c_{\mathrm{max}})$ post-swap rounds, for a cost of $O(L^3\,\mathcal{C}_{\mathrm{max}}^2\log c_{\mathrm{max}})$ plus the once-per-link frontier construction, where $\mathcal{C}_{\mathrm{max}}=\max_{i,j}|\mathcal{C}_{i,j}|$ is the largest retained Pareto set; it admits no useful a priori bound and is reported alongside runtime in Section~\ref{sec:results}.



\section{Numerical Results}\label{sec:results}

We evaluated the framework on chains whose source fidelities are drawn independently per link as $\mu_e\sim\mathcal{N}(\bar{\mu},\sigma^2)$---heterogeneous across links, uniform within each---with $\bar{\mu}=0.85$ and $\sigma=0.1$, with baseline (raw) fidelities $F_e^{(0)}$ given by Eq.~\eqref{eq:baseline-fidelity} and clipped to $[0.5,0.999]$. Each plotted point aggregates $200$ instances generated from a common seed, so all algorithms face identical draws. Unless swept, the default parameters are $L=4$, $c_e=50$ ($R_{\mathrm{gen},e}=5\times10^{3}$ pairs/s at $\tau_{\mathrm{gen}}=10$~ms), and $F_{\mathrm{th}}=0.70$. We evaluate the framework under both the symmetric Pauli noise model (which leads to Werner states) and the single-Pauli-error model (which leads to two-Bell-state diagonal states); while Section~\ref{sec:results_s2} compares both, our primary results focus on symmetric Pauli noise as the more physically representative regime. Two analytical PtS bounds are verified in the results: the capacity-independent pumping bound of Eq.~\eqref{eq:pts_ceiling} and the capacity-dependent tree bound of Eq.~\eqref{eq:tree_ceiling}.

We compared six algorithms: \emph{MILP (pumping)} and \emph{MILP (tree)}, the two instantiations of Section~\ref{sec:milp} with the pumping and tree choice sets; \emph{DP (pumping)} and \emph{DP (tree)}, the DP of Section~\ref{sec:dp} with the corresponding choice sets; \emph{Heuristic (tree)}, the greedy heuristic of Section~\ref{sec:heuristic} over tree frontiers; and two published baselines --- \emph{Q-PATH}~\cite{li2022_qpath_qleap}, ported with its published selection rule intact and scored, like the other methods, with the objective of Eq.~\eqref{eq:objective}, and \emph{Chen--Jia per-link}~\cite{chen2024_jsac}, which schedules each link optimally but coordinates no width across links. Q-PATH selects a purification depth per link, as MILP (pumping) does, and is therefore compared against it in the fidelity-threshold sweep (Section~\ref{sec:results_s1}); the other five methods appear in every sweep. We cap the depth of the pumping-based methods at $k_{\max}=20$; the cap is never met. 
Because the MILPs optimize logarithmic quantities at finite solver precision, every returned selection is recomposed with the exact equations of Section~\ref{sec:quantum_protocols}; a selection that fails the threshold is excluded by a no-good cut and the MILP is re-solved, so all methods answer to one exact feasibility test.

We report two metrics per operating point. The \emph{blocking probability} is the fraction of the $200$ instances for which no feasible routing exists ($T=0$). The \emph{median served throughput} is the median of $T$ over the instances the method serves, on a logarithmic axis with values below $10^{-3}$ drawn at the marked display floor. Blocking is the primary metric --- a blocked request delivers nothing, a served request of low throughput still delivers pairs --- so we read the throughput results as the price of the feasibility gains; Section~\ref{sec:results_s1} quantifies the one caveat, a served instance falling below an operator's minimum rate. This priority also explains an apparent reversal relative to earlier studies of operation order: under the single-Pauli-error noise model pumping has no fidelity ceiling, so every algorithmic difference must surface in throughput, whereas under symmetric Pauli noise feasibility itself is the contested resource.

\begin{figure*}
    \centering
    \includegraphics[width=\linewidth,height=8.5cm]{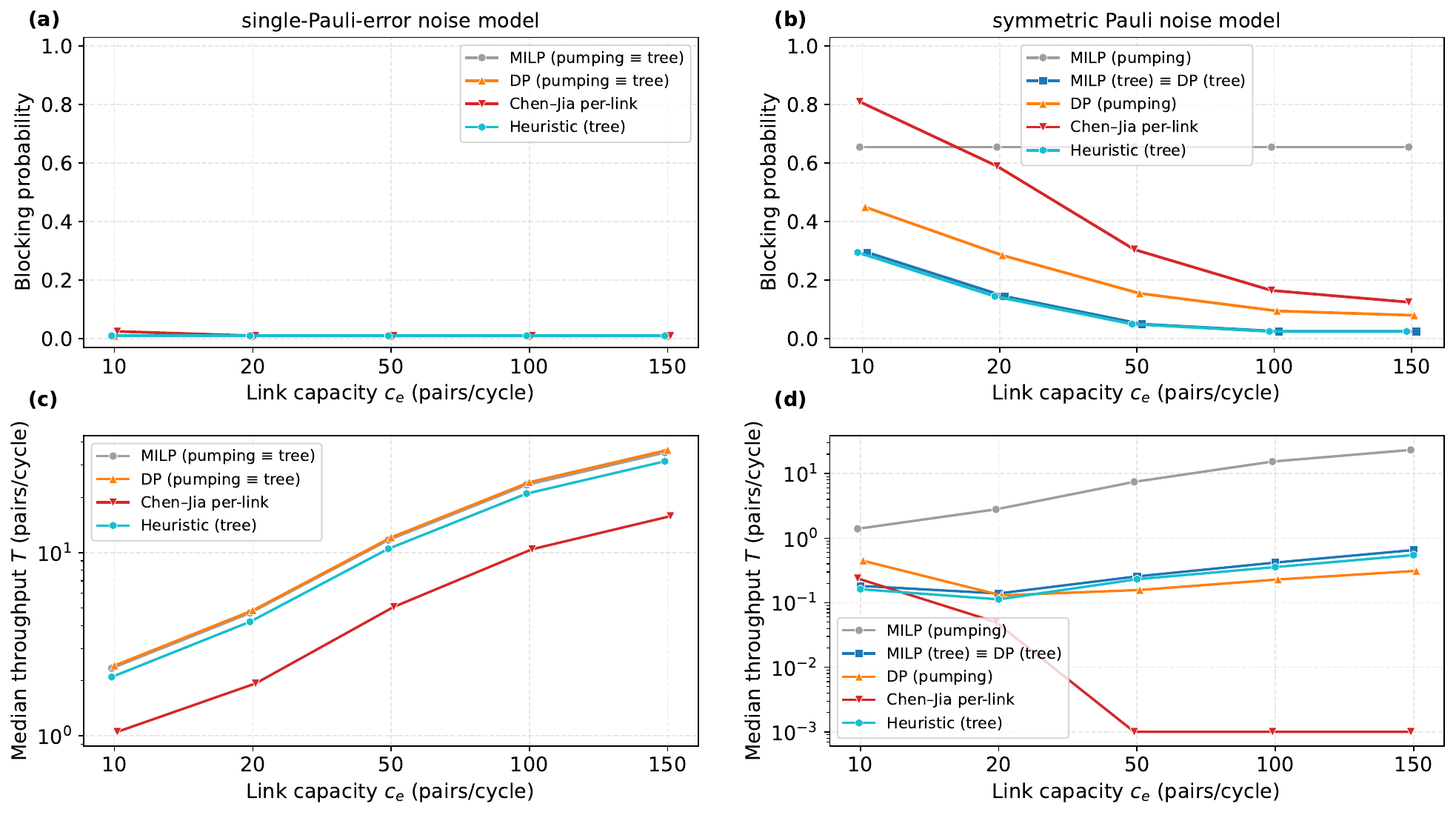}
    \caption{Capacity sweep at the default $L=4$, $F_{\mathrm{th}}=0.70$, under the single-Pauli-error noise model (left) and symmetric Pauli noise model (right): (a),(b) blocking probability; (c),(d) median served throughput (log scale, display floor $10^{-3}$). Merged legend entries ($\equiv$) denote algorithms whose blocked instance sets coincide at every point; in the single-Pauli column each method's pumping and tree instantiations coincide exactly.}
    \label{fig:s2}
\end{figure*}

\subsection{The Two Noise Models}\label{sec:results_s2}

Fig.~\ref{fig:s2} sweeps the capacity from $c_e=10$ to $150$ under both noise models. Under the single-Pauli-error noise model (Fig.~\ref{fig:s2}(a),(c)) the blocking differences between algorithms vanish. The shape of the purification tree is irrelevant under this model (Section~\ref{sec:schedules}), so each method's pumping and tree instantiations coincide --- verified per instance across all $1000$ (five capacities, $200$ instances each) --- and blocking is a flat $\sim1\%$ for every method and capacity. This residual is the two instances in which a link's baseline fidelity was clipped to the lower bound $F_{e}^{0}=1/2$ : such a pair carries no entanglement and purification cannot improve it, so no method serves these instances. Operation order adds nothing measurable either, because the pumping recursion approaches unit fidelity doubly exponentially (Proposition~\ref{prop:fixed_points}\textit{(i)}) and leaves post-swap purification nothing to rescue. To check that this reflects the model rather than the operating point, we raised $F_{\mathrm{th}}$ toward $1$ on a $20$-instance probe: blocking beyond the clipping residual begins only above $F_{\mathrm{th}}=0.999$, and even there DP (pumping $\equiv$ tree) served two probe instances that the PtS optimum could not --- a thin improvement shell of Section~\ref{sec:schedules}. In throughput (Fig.~\ref{fig:s2}(c)) the two exact methods coincide, Heuristic (tree) trails them slightly, and Chen--Jia per-link lies well below, since it schedules each link in isolation and never coordinates width across links. Larger throughput differences between operation orders would require regimes where the required depth competes with capacity (lower baseline fidelities, tighter thresholds, smaller capacities). What this model cannot produce at any operating point is a difference in feasibility.

Under symmetric Pauli noise (Fig.~\ref{fig:s2}(b),(d)) the differences appear. In blocking (Fig.~\ref{fig:s2}(b)), MILP (pumping) blocks an identical $66\%$ at all five capacities, equal to Eq.~\eqref{eq:pts_ceiling}: the bound contains no $c_e$, and no increase in capacity lifts the pumping class over it. The blocking of MILP (tree) $\equiv$ DP (tree) falls from $0.30$ at $c_e=10$ to $0.03$ at $c_e=150$, tracking Eq.~\eqref{eq:tree_ceiling}, with Heuristic (tree) alongside it. DP (pumping) lies between MILP (pumping) and the tree methods at every capacity, since the free operation order recovers part of the tree advantage. The Chen--Jia per-link blocks the most at small capacity and improves steadily with it. Much of what capacity buys the tree methods, however, has low throughput: the instances they serve at $T<10^{-3}$ grow from $3\%$ to $24\%$ of the total, so with a rate requirement of $10^{-3}$ their blocking would fall only from $0.325$ to $0.265$. Capacity mostly converts infeasible requests into feasible ones of low throughput (Fig.~\ref{fig:fan} examines such requirements). In throughput (Fig.~\ref{fig:s2}(d)), MILP (pumping) shows the highest median at every capacity because it serves only the easiest third of the instances (Section~\ref{sec:same_set_throughput} makes the comparison fair). The tree methods and Heuristic (tree) lie close together. DP (pumping) and Chen--Jia per-link lie above them at $c_e=10$, where they too serve only easier instances, and below them from $c_e=20$; Chen--Jia per-link reaches the display floor from $c_e=50$. 

The single-Pauli-error noise model, adopted in the literature for its closed-form composition, thus eliminates most of the effects under study. For this reason, the remaining sweeps use symmetric Pauli noise.

\begin{figure}
    \centering
    \includegraphics[width=\linewidth]{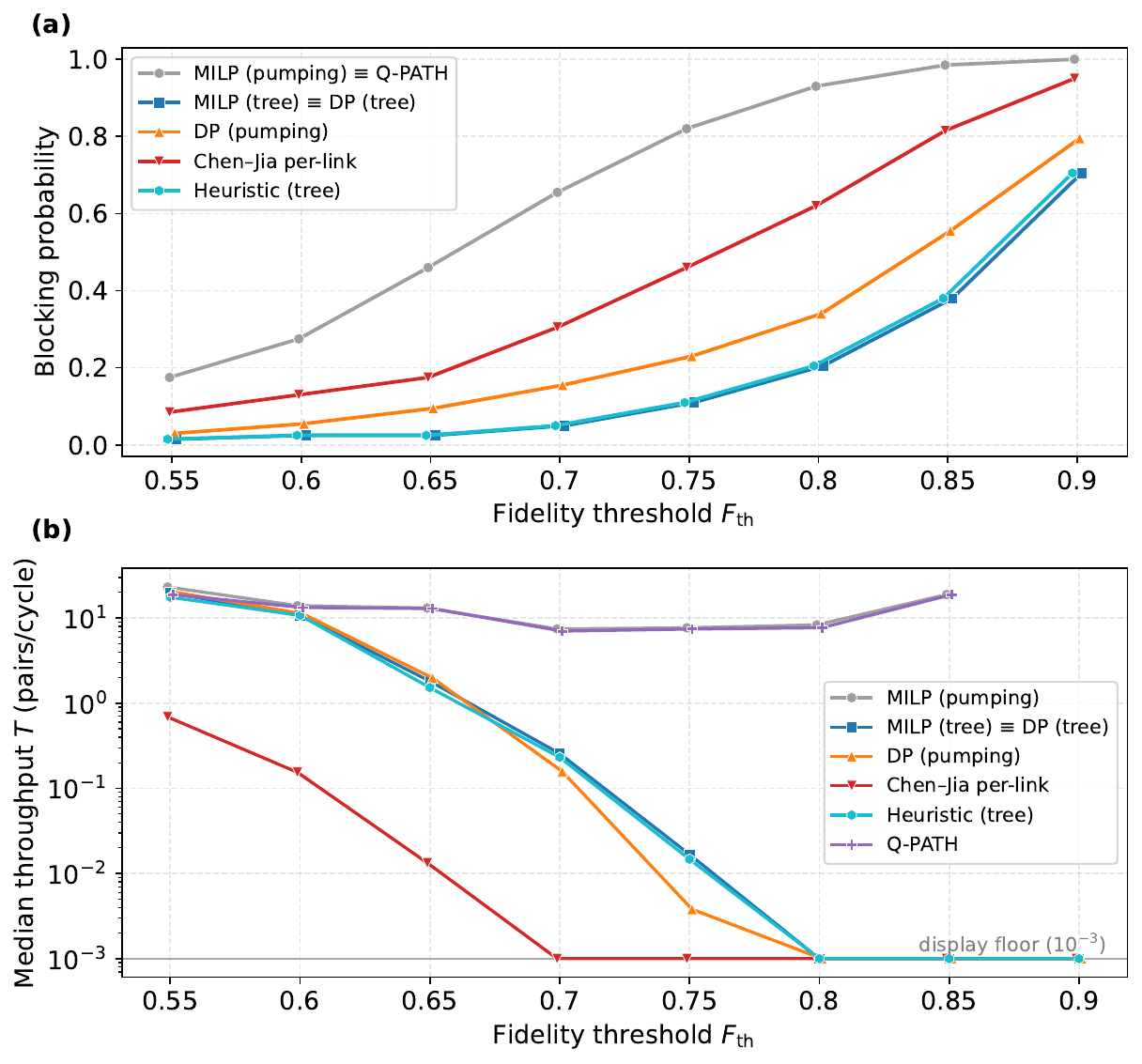}
    \caption{Fidelity-threshold sweep at the default $L=4$, $c_e=50$, symmetric Pauli noise: (a) blocking probability; (b) median served throughput. \mbox{MILP (pumping) $\equiv$ Q-PATH} in (a): identical blocked sets at all eight thresholds --- blocking is a property of the pumping class that no selection rule can move; in (b) Q-PATH's selection rule costs it throughput. \mbox{MILP (tree) $\equiv$ DP (tree)}: identical blocked sets at every threshold.}
    \label{fig:s1}
\end{figure}

\begin{figure}
    \centering
    \includegraphics[width=\linewidth, height=4cm, trim={0.35cm 0 0 0}, clip]{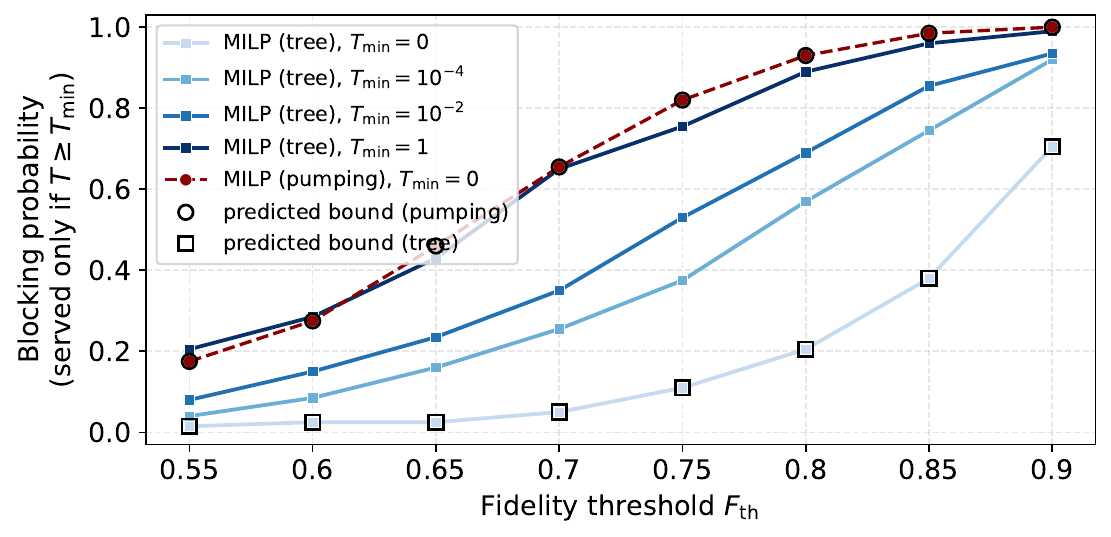}
    \caption{Blocking of MILP (tree) when a served instance must also meet a minimum rate $T_{\min}\in\{0,10^{-4},10^{-2},1\}$ pairs per cycle (data of Fig.~\ref{fig:s1}); hollow markers: the bounds of Eqs.~\eqref{eq:pts_ceiling} and~\eqref{eq:tree_ceiling}. MILP (pumping), dark red, is invariant up to $T_{\min}=10^{-2}$.}
    \label{fig:fan}
\end{figure}

\subsection{Fidelity-Threshold Sweep}\label{sec:results_s1}

Fig.~\ref{fig:s1} sweeps $F_{\mathrm{th}}\in[0.55,0.90]$ under symmetric Pauli noise. In Fig.~\ref{fig:s1}(a), the blocking of MILP (pumping) rises from $0.17$ to $1$ and coincides with Eq.~\eqref{eq:pts_ceiling} at every point. Q-PATH blocks the identical instance set at all eight thresholds and is drawn as the same curve. The reason is simple: both methods choose one purification depth per link, so both explore the same set of depth vectors; whether a feasible vector exists at all is the same question for both, and their selection rules decide only which feasible vector is returned --- at best, the MILP's optimum. Tree schedules move the boundary: MILP (tree) $\equiv$ DP (tree) blocks from $0.015$ to $0.705$; of the instances that MILP (pumping) cannot serve, it serves more than $90\%$ at thresholds up to the default $0.7$, and still $29\%$ at $0.90$, where the tree bound of Eq.~\eqref{eq:tree_ceiling} engages. That DP (tree) blocks the identical set at every threshold is the certificate anticipated in Section~\ref{sec:schedules}: post-swap purification, the only capability the DP adds, contributes no feasibility beyond the frontier. The blocking of DP (pumping) lies between MILP (pumping) and MILP (tree), from $0.03$ to $0.80$: with pumping on the links but a free operation order, it recovers much of the tree advantage without a tree on any link (Section~\ref{sec:results_protocols}). Heuristic (tree) matches the served set of the exact tree optimum at all eight thresholds, within $16\%$ of its median throughput for $F_{\mathrm{th}}\le0.80$, in milliseconds.

Fig.~\ref{fig:s1}(b) shows the price. The tree methods, DP (pumping), and the heuristic fall steeply as the threshold rises, reaching the display floor by $0.80$, because they serve the harder instances near the pumping bound, whose routing tree spend width on depth. MILP (pumping) and Q-PATH, in contrast, stay near $10$ pairs per cycle throughout: they serve only the easy instances, and at $0.85$ only three of them (Section~\ref{sec:same_set_throughput} removes this population effect). As expected, Q-PATH's selection rule costs it a little throughput against MILP (pumping) on jointly served instances, and Chen--Jia per-link loses on heterogeneous draws where the bottleneck link is not the lowest-fidelity one.

A served instance can still fail to be useful if its throughput falls below a required minimum rate $T_{\min}$. An instance thus fails in one of two ways: it is infeasible ($T=0$), or it is feasible but served at $T<T_{\min}$. Fig.~\ref{fig:fan} re-evaluates the blocking of Fig.~\ref{fig:s1}(a), as a function of $F_{\mathrm{th}}$, counting both failures for several values of $T_{\min}$. For the pumping class the two failures coincide: whatever fidelity pumping reaches, it reaches at shallow depth and hence at nearly full width, so the MILP (pumping) curve does not move for any $T_{\min}\le10^{-2}$. The tree class separates them: the instances that pumping cannot serve are served by trees only through deep schedules, each additional tree level halving the width and squaring the success probability, so such an instance may be feasible only at width $1$ and with low success probability. The MILP (tree) curves therefore fan upward as the requirement tightens. The tree advantage of Fig.~\ref{fig:s1}(a) persists for $T_{\min}\le10^{-2}$. At $T_{\min}=1$ --- one delivered pair per cycle, i.e., $100$ pairs/s at $\tau_{\mathrm{slot}}=10$~ms --- the tree curve nearly coincides with the strict pumping curve: this requirement erases most of what tree schedules gain in feasibility.

\begin{figure}
    \centering
    \includegraphics[width=\linewidth,height=11.8cm]{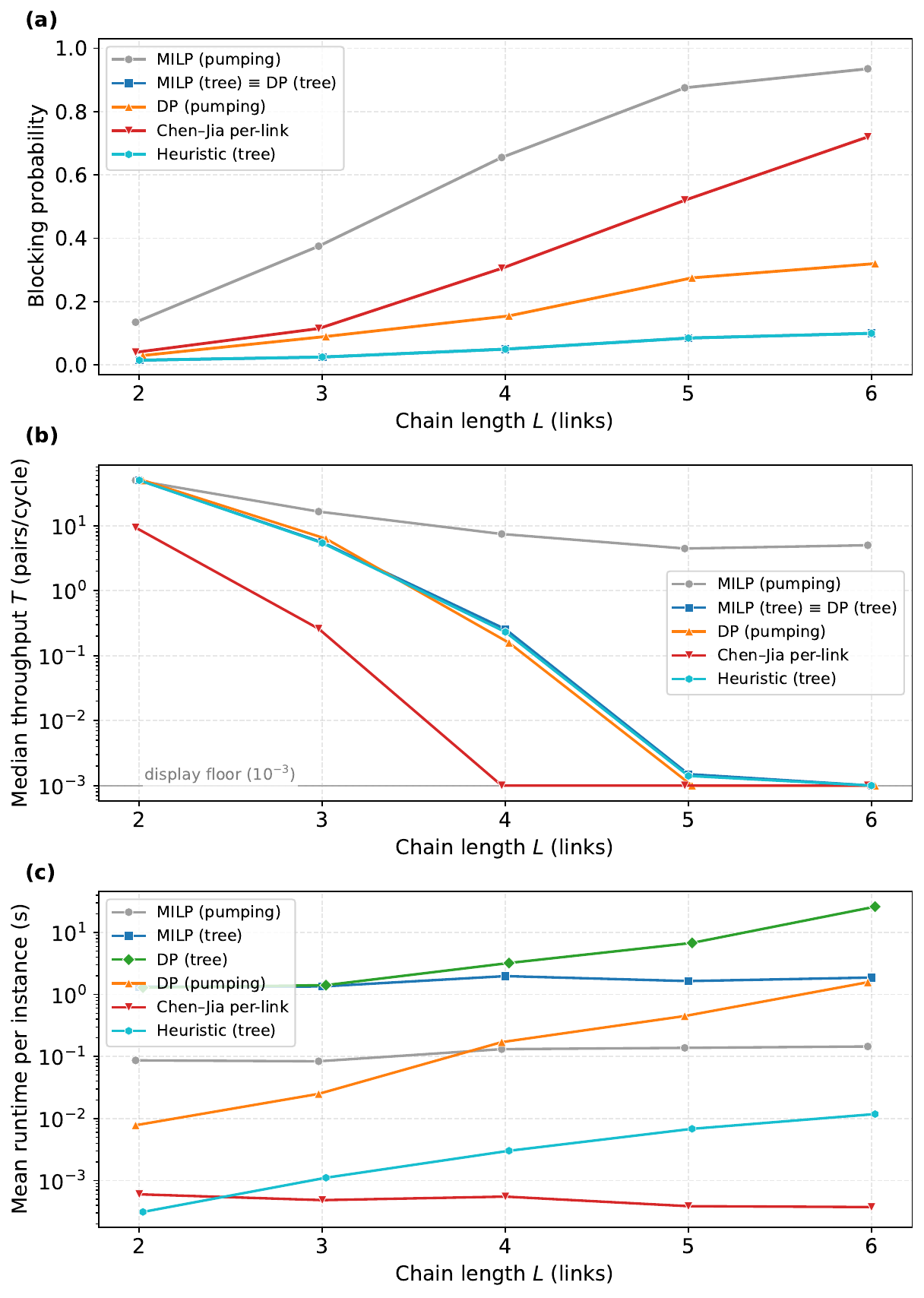}
    \caption{Chain-length sweep at the default $c_e=50$, $F_{\mathrm{th}}=0.70$, symmetric Pauli noise: (a) blocking probability; (b) median served throughput; (c) per-instance runtime (log scale, including the once-per-instance frontier construction for tree methods). MILP (tree) and DP (tree) appear separately in (c): the difference is cost, not solution quality.}
    \label{fig:s3}
\end{figure}

\subsection{Chain Length and Runtime}\label{sec:results_s3}

Fig.~\ref{fig:s3} sweeps $L\in\{2,\ldots,6\}$. Both analytical bounds tighten as Bell parameters multiply across links (Fig.~\ref{fig:s3}(a)): the blocking of MILP (pumping) rises to $0.66$ at $L=4$ and $0.94$ at $L=6$, which defines the maximum chain length of Section~\ref{sec:schedules} for every PtS method of the pumping class, while the blocking of MILP (tree) $\equiv$ DP (tree) rises only from $0.015$ to $0.10$, with identical served sets at every length. The merged MILP (tree) $\equiv$ DP (tree) curve conceals one genuine difference: in Fig.~\ref{fig:s3}(b), DP (tree) exceeds MILP (tree) in median throughput by up to $\sim10\%$ at $L\geq 4$, through mixed routing trees that meet the threshold with shallower link schedules and one post-swap round --- the operation order buying efficiency, not feasibility. Median served throughput itself decays by seven orders of magnitude from $L=2$ to $L=6$, the compounded price of multiplicative success probabilities.

Fig.~\ref{fig:s3}(c) shows the mean runtime per instance. The runtimes of the two DPs grow with the chain length, as their retained Pareto sets grow: DP (pumping) from about $10$~ms to about $2$~s, and DP (tree) from about $1$~s to about $25$~s at $L=6$, where its Pareto sets reach $\mathcal{C}_{\max}=42{,}688$ configurations (one of the $200$ instances exceeded a $6$~GB per-instance memory guard and is excluded from its statistics). The two MILPs stay nearly flat: MILP (pumping) at about $0.1$~s, and MILP (tree) at $1$--$2$~s, since their variables depend on the frontier size. Heuristic (tree) is the fastest tree method, below $15$~ms at every length, which makes it the natural candidate for online operation. DP (tree), needed only to verify that a free operation order adds nothing, is affordable offline and unnecessary online. At $L=8$ DP (tree) exceeds the $32$~GB of the machine used; we disclose the limit rather than report partial runs, and no conclusion depends on that case.

\subsection{Throughput on a Common Instance Set}\label{sec:same_set_throughput}

The median throughput of each method is computed over the instances that the method serves --- a different set per method --- and this hides a population effect: MILP (pumping) serves fewer instances than every other method (Q-PATH excepted, which serves the same set), so the other medians include harder instances it cannot serve, which lowers their throughput. We therefore add a comparison on common ground, the instances the MILP (pumping) serves. On this set every other method returns a solution at least as good as MILP (pumping) on every instance --- the pumping choice set is contained in the tree frontier, and PtS is one of the orders the DP searches --- and we verified this dominance on every run. Fig.~\ref{fig:Tratio} plots the ratio of each method's aggregate throughput to that of MILP (pumping) on the common set for the three sweeps. The gains are real but modest, and they grow with difficulty: under 1\% at loose thresholds and short chains, a few percent near the feasibility boundary, and insensitive to capacity. Because the common set shrinks as the threshold rises, the three highest thresholds in Fig.~\ref{fig:Tratio}(a) use 2000-instances --- which leave at least 47 common instances per point; the open marker at $L=6$ in Fig.~\ref{fig:Tratio}(c) flags fewer than 20 instances. The lower tree medians of the other figures thus reflect no worse solution on any instance, only the harder instances the tree methods additionally serve.

\begin{figure}[t]
    \centering
    \includegraphics[width=\linewidth,height=11.8cm]{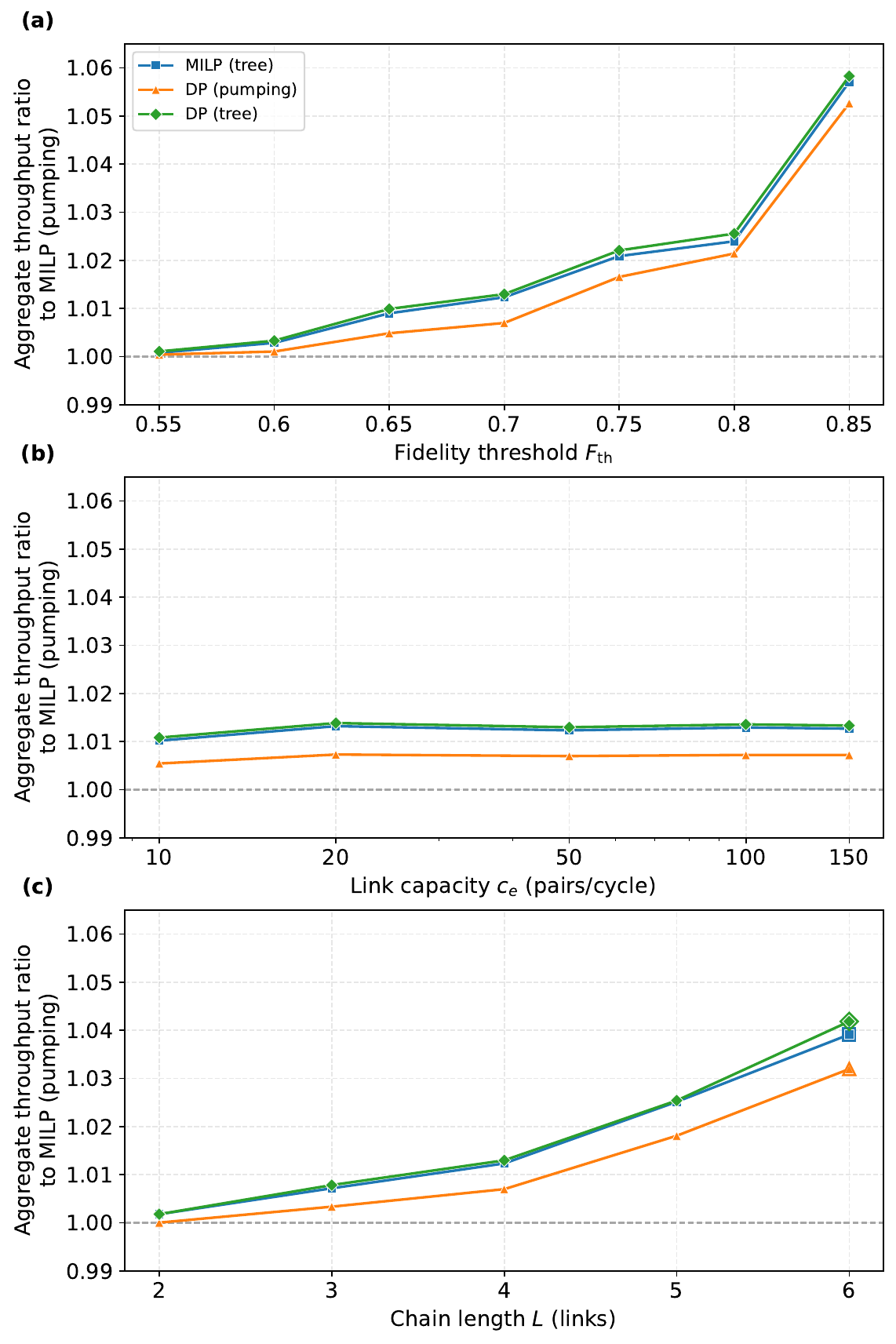}
    \caption{Aggregate throughput relative to MILP (pumping) on the instances MILP (pumping) serves, against (a) the fidelity threshold, (b) the link capacity, and (c) the chain length; other parameters at defaults. In (a), $F_{\mathrm{th}}\geq 0.75$ uses a 2000-instance extension of the same seed (a superset of the 200 used elsewhere); open marker: fewer than 20 common instances.}
    \label{fig:Tratio}
\end{figure}

\subsection{Inside the Routing Trees}\label{sec:results_protocols}

The routing trees returned by the DP show what the free operation order is used for. Within the pumping class, the fraction of served instances whose optimal routing interleaves --- purify, swap, purify the composed pair --- rises from $13\%$ at $F_{\mathrm{th}}=0.55$ to $100\%$ at $0.85$: near the bound, every request strict PtS cannot serve is served by a mixed routing tree. A representative instance at $F_{\mathrm{th}}=0.70$ makes the mechanism concrete: MILP (pumping)'s best allocation serves at $T=0.12$, while DP (pumping) purifies the four links to depths $(2,2,3,2)$, swaps, and applies one post-swap round, meeting the threshold at $T=0.84$ --- a factor of $7.0$ --- by spending the composed pair's width where per-link pumping has saturated.

Read against the tree results, these gains admit a sharper interpretation. Post-swap purification is the self-purification recursion applied after composition (Section~\ref{sec:schedules}); for links restricted to pumping it is the only route to that recursion, and the large ordering gains above --- like those reported empirically in prior work~\cite{victora2023_purification,haldar2025_quasilocal,zang2023_buffer} --- are, on this evidence, chiefly indirect access to the tree class rather than a benefit of the operation order itself. Once links schedule trees directly, the same freedom adds no feasibility at any operating point of Figs.~\ref{fig:s2}--\ref{fig:s3} and at most ${\sim}10\%$ median throughput; its residual feasibility value is a thin shell at the bound. At the scales measured here, operation order and purification schedule are substitutes, and the practical conclusion is the cheaper one: a millisecond tree-frontier MILP, or its heuristic, captures what DP and interleaved routing trees were reaching for.


\section{Conclusions}\label{sec:conclusions}
 
We investigated capacity-aware entanglement routing on quantum repeater chains, developing a unified optimization framework under finite link capacities. By combining analytical feasibility bounds with exact MILP and DP formulations and a low-complexity heuristic, the framework enables systematic optimization of quantum-resource allocation across schedule classes, operation orders, and physical noise models.

Our analysis reveals that the underlying noise model fundamentally affects the attainable fidelity and the value of different routing strategies. Under symmetric Pauli noise, pumping imposes a capacity-independent fidelity ceiling that can be overcome through general purification trees or post-swap purification. Numerical results confirm the analytical bounds and demonstrate the importance of exploiting the available purification choices: tree schedules served over 90\% of pumping-infeasible requests at moderate thresholds, while the tree-based MILP and order-exact DP served identical instance sets across all evaluated operating points. Under single-Pauli-error noise, the distinction between purification schedule classes disappears in terms of attainable fidelity.

Importantly, the proposed framework is not tied to a particular purification schedule, operation order, or noise model. Its formulations exploit common structural properties of entanglement-processing operations, allowing noise models and routing objectives that satisfy those properties to be incorporated through model-dependent parameters without redesigning the optimization algorithms. 

Taken together, our findings show that the structure of purification ---not merely its depth or the number of available entangled pairs--- is a central consideration in capacity-aware quantum-network design. They also suggest that, before introducing more complex operation-order optimization, it is important to understand which fidelity limitations can already be overcome through appropriate link-level scheduling. Next steps include the extension of the framework from chains to general network topologies, including path selection and competition for shared resources. Taken together, the results suggest treating the purification schedule and the operation order as distinct design dimensions: the schedule governs feasibility, and the operation order matters only where the schedules alone fall short.

\bibliographystyle{IEEEtran}
\bibliography{references.bib}

\renewcommand{\appendixname}{Appendix}
\appendices


\section{Bell States, Noise Models, and Twirling}
\label{app:quantum}

\subsection{The Bell basis}

The four Bell states $\{\ket{\Phi^\pm},\ket{\Psi^\pm\}}$ can be generated from the
computational basis as
\begin{equation}
\ket{\beta_{n_1,n_2}}
=\mathrm{CNOT}(H\otimes I)\ket{n_1n_2},
\label{Bell}
\end{equation}
where $n_1,n_2\in\{0,1\}$, $H$ is the Hadamard gate, and
$\mathrm{CNOT}$ is the controlled-NOT gate. Writing
$\ket{\beta_{n_1,n_2}}\equiv\ket{\beta_{2n_1+n_2}}$, we have
$\ket{\beta_0}=\ket{\Phi^+}$,
$\ket{\beta_1}=\ket{\Psi^+}$,
$\ket{\beta_2}=\ket{\Phi^-}$, and
$\ket{\beta_3}=\ket{\Psi^-}$.

Ideally, the two nodes connected by an elementary link share the target Bell state $\ket{\beta_0}$, also known as an EPR pair. In practice,
channel noise, imperfect state preparation, and memory decoherence produce
a mixed state
$\rho=\sum_{n,j}\lambda_{n,j}\ket{\beta_n}\bra{\beta_j}$,
where $\Lambda=(\lambda_{n,j})$ is positive semidefinite with unit trace.
Its fidelity with respect to the target Bell state is
$F(\rho)=\bra{\beta_0}\rho\ket{\beta_0}$.

\subsection{Noise models}

In this work, we adopt the common assumption of Pauli noise
\cite{nielsen2010quantum}, described in terms of the single-qubit
Pauli operators $\{X,Y,Z\}$.

\subsubsection{Asymmetric Pauli errors -- Bell-diagonal states}
We first consider the general case in which a Pauli error acts on one
qubit of the entangled pair. The resulting state can be written as
\begin{equation}
\rho=\sum_{G\in\{I,X,Y,Z\}}
p_G\,(I\otimes G)\ket{\beta_0}\bra{\beta_0}(I\otimes G)^{\dagger},
\label{anoise}
\end{equation}
where $p_G\geq0$ and $\sum_G p_G=1$.
Since the Pauli operators map Bell states onto Bell states,
Eq.~\eqref{anoise} yields a Bell-diagonal state
\begin{equation}
\rho_B(\vec{p})
= p_I\ket{\beta_0}\bra{\beta_0}
+ p_X\ket{\beta_1}\bra{\beta_1}
+ p_Z\ket{\beta_2}\bra{\beta_2}
+ p_Y\ket{\beta_3}\bra{\beta_3},
\label{bel}
\end{equation}
with $\vec p=(p_I,p_X,p_Z,p_Y)$ and fidelity
$F=\bra{\beta_0}\rho_B\ket{\beta_0}=p_I$.


\subsubsection{Symmetric Pauli Noise Model -- Werner states}
In this model,  the three nontrivial Pauli errors occur with
equal probability, such that $p_I=F$ and
$p_X=p_Y=p_Z=(1-F)/3$. The resulting state is a Werner state, given by
\begin{equation}
\rho_W(F)
= F\ket{\beta_0}\bra{\beta_0}
+ \frac{1-F}{3}\sum_{j=1}^{3}\ket{\beta_j}\bra{\beta_j}.
\label{Wer}
\end{equation}


\subsubsection{Single-Pauli-Error Model --- Two-Bell-State Diagonal States}
In this model, only one type of nontrivial Pauli error occurs.
Taking this error to be $X$, we have $p_I=F$, $p_X=1-F$, and
$p_Y=p_Z=0$.  The resulting state reduces to a mixture of two Bell states,
\begin{equation}
\rho_{2B}(F)
= F\ket{\beta_0}\bra{\beta_0}
+ (1-F)\ket{\beta_1}\bra{\beta_1}.
\label{2Bell}
\end{equation}


The symmetric Pauli and single-Pauli-error noise models represent different
physical noise regimes. Symmetric Pauli noise provides an effective
description when the accumulated imperfections have no strongly preferred
Pauli component, so that the error is distributed among the three
nontrivial Pauli operators. The resulting Werner state therefore provides
a natural model for approximately isotropic noise. In contrast, the
two-Bell-state diagonal states represent a strongly asymmetric noise regime in
which single Pauli error type dominates. Considering both models allows us
to distinguish the effect of the overall fidelity from that of the
underlying error structure.

\subsection{Bell parameter and operation noise}

For both noise models used in this work, the Bell parameters introduced in
Eq.~\eqref{eq:polarization} provide a convenient representation of the
corresponding state families:
\begin{equation}
\begin{split}
&\rho_W(F)
= \xi_W\ket{\beta_0}\bra{\beta_0}
+(1-\xi_W)\frac{I}{4},\\
&\rho_{2B}(F)
= \xi_{2B}\ket{\beta_0}\bra{\beta_0}
+(1-\xi_{2B})\frac{\Pi_2}{2},
\end{split}
\label{eq:normal_forms}
\end{equation}
where $\Pi_2=\ket{\beta_0}\bra{\beta_0}+\ket{\beta_1}\bra{\beta_1}$,
and
\begin{equation}
\xi_W(F)=\frac{4F-1}{3},
\qquad
\xi_{2B}(F)=2F-1.    
\end{equation}
 Concurrence is an entanglement measure for two-qubit states,
ranging from zero for separable states to unity for maximally
entangled states \cite{Wootters1998}.
In the entangled regime, $\xi_{2B}$ coincides with the concurrence
of the two-Bell-state diagonal state, whereas for Werner states
the concurrence is $C=\max\{0,(3\xi_W-1)/2\}$.
Most importantly for our analysis, entanglement swapping has a simple
multiplicative form when expressed in terms of $\xi$.

Imperfections in quantum operations can be incorporated through an additional
reduction of the Bell parameter, $\xi\mapsto\eta\xi$, where
$0\leq\eta\leq1$ characterizes the operation quality. This representation
is exact when the operation noise preserves the corresponding state family.
For the Werner family, symmetric Pauli noise preserves the  family and
acts as $\xi_W\mapsto\eta\xi_W$. Similarly, for the two-Bell-state diagonal  states,
single-Pauli-error noise preserves the  family and can be
represented as $\xi_{2B}\mapsto\eta\xi_{2B}$. Thus, the noise associated
with quantum operations must be compatible with the assumed state family
if the same one-parameter description is to be preserved. In particular,
the two-Bell-state diagonal family is not closed under symmetric Pauli noise,
which distributes population over all four Bell states, whereas it is
closed under the corresponding single-Pauli-error noise.


\subsection{Isotropic twirling}

The BBPSSW purification protocol (see Sec.~\ref{sec:ent_models_and_protocols})
is formulated for Werner input states, Eq.~\eqref{Wer}, but its output
is generally Bell-diagonal, Eq.~\eqref{bel}, rather than Werner.
Repeated application therefore requires transforming the output back to
the Werner family before each purification round. This is achieved by
the isotropic twirling
\begin{equation}
\mathcal{T}(\rho)
= \int dU\,(U\otimes U^{*})\,\rho\,
(U^{\dagger}\otimes U^{T}),
\label{eq:twirl}
\end{equation}
where the integral is taken over the Haar measure of $SU(2)$ and $U^{*}$
denotes complex conjugation in the computational basis. Because
$\ket{\beta_0}$ is invariant under $U\otimes U^{*}$, the twirl preserves
the target-state population while averaging the remaining Bell-state
populations, mapping Eq.~\eqref{bel} onto Eq.~\eqref{Wer} with $F=p_I$.
Twirling therefore preserves the fidelity. Moreover, when
$p_I=p_{\max}$, it also preserves the concurrence of the Bell-diagonal
state.

The states of Eq.~\eqref{Wer} are the $U\otimes U^{*}$-invariant, or
isotropic, states. They are related by a local unitary to the
$U\otimes U$-invariant states originally introduced by Werner, and the
quantum-networking literature commonly refers to both families as Werner
states; we follow that convention. In practice, the twirling can be
implemented by averaging over a finite set of bilateral unitaries rather
than over the continuous Haar measure.


\section{Structural Properties Underlying the Optimization Framework}
\label{app:convexity}

The optimization framework of Section~\ref{sec:algorithms} (MILP, heuristic, DP) apply to both the symmetric Pauli (Werner quantum states) and to the single-Pauli-error noise models (two-Bell-state diagonal quantum states) because they rely on three structural properties: \textit{(i)} diminishing returns of purification (used only by the heuristic), \textit{(ii)} log-linearity of swap composition in the Bell parameter, and \textit{(iii)} objective compatibility. This appendix proves Proposition~\ref{prop:fixed_points} and verifies these properties for both noise models; any noise model and objective satisfying them can be incorporated without modifying the framework.

Throughout we use the Bell parameters $\xi_W(F)=(4F-1)/3$ and $\xi_{2B}(F)=2F-1$ introduced in Eq.~\eqref{eq:polarization} and discussed in Appendix~\ref{app:quantum}. Both lie in $[0,1]$ in the entangled regime where ($F>1/2$). We write $\xi$ when a statement holds for either state family.

\subsection{Property 1. Purification Limits and Strict Concavity}
\label{app:convexity:concavity}
Consider pumping purification on a link with raw-pair baseline fidelity $F^{(0)}$ (dropping link indices for readability). Let $F_{\mathrm{pump}}^{(k)}$ denote the fidelity after $k$ pumping rounds. For the single-Pauli-error noise model, Eq.~(\ref{2B_pur}) gives the pumping recursion
\begin{equation}
F_{\mathrm{pump}}^{(k+1)}
=\frac{F_{\mathrm{pump}}^{(k)} F^{(0)}}
{F_{\mathrm{pump}}^{(k)} F^{(0)} + \left(1-F_{\mathrm{pump}}^{(k)}\right)\left(1-F^{(0)}\right)},
\label{eq:bbpssw_twobell}
\end{equation}
and for the symmetric Pauli noise model, Eq.~(\ref{W_pur}) gives
\begin{equation}
F_{\mathrm{pump}}^{(k+1)}=
\frac{1-F_{\mathrm{pump}}^{(k)}-F^{(0)}+10F_{\mathrm{pump}}^{(k)}F^{(0)}}
{5-2F_{\mathrm{pump}}^{(k)}-2F^{(0)}+8F_{\mathrm{pump}}^{(k)}F^{(0)}}.
\label{eq:bbpssw_werner}
\end{equation}

\subsubsection{Proof of Proposition~\ref{prop:fixed_points}}
\quad \textit{(i)} Defining the error ratio $x_0=(1-F^{(0)})/F^{(0)}\in(0,1)$, Eq.~\eqref{eq:bbpssw_twobell} admits the closed-form solution
\begin{equation}
F_{\mathrm{pump}}^{(k)}=\frac{1}{1+x_0^{\,k+1}},
\label{eq:2b_closed_form}
\end{equation}
which increases strictly to unity as $k\to\infty$.

\noindent\textit{(ii)} For symmetric Pauli errors, the fixed points of Eq.~\eqref{eq:bbpssw_werner} satisfy $F_{\mathrm{pump}}^{(k+1)}=F_{\mathrm{pump}}^{(k)}=F$, yielding the quadratic equation
\begin{equation}
\left(8F^{(0)}-2\right)F^{2}
+\left(6-12F^{(0)}\right)F
+\left(F^{(0)}-1\right)=0.
\label{eq:fixed_point}
\end{equation}
The constant term $F^{(0)}-1$ is strictly negative for $F^{(0)}<1$ and the square coefficient $8F^{(0)}-2$ is strictly positive for $F^{(0)}>1/2$, so the quadratic has exactly one positive root:
\begin{equation}
\begin{aligned}
&F_{\mathrm{pump}}^{\ast}\!\left(F^{(0)}\right)
=\frac{12F^{(0)}-6}{2\left(8F^{(0)}-2\right)} \\
&\quad+\frac{
\sqrt{\left(6-12F^{(0)}\right)^{2}
+4\left(8F^{(0)}-2\right)\left(1-F^{(0)}\right)}
}{
2\left(8F^{(0)}-2\right)
}.
\end{aligned}
\label{eq:fixed_point_closed}
\end{equation}
Because $\mathrm{F}_{\mathrm{pur}}(F^{(0)},F^{(0)})>F^{(0)}$ and $\mathrm{F}_{\mathrm{pur}}(1,F^{(0)})<1$, this root lies strictly in $(F^{(0)},1)$. Since $\mathrm{F}_{\mathrm{pur}}$ is continuous and strictly increasing, $F_{\mathrm{pump}}^{(k)}$ converges monotonically to $F_{\mathrm{pump}}^{\ast}$.

\noindent\textit{(iii)} The balanced-tree fidelity of Section~\ref{sec:schedules} follows $F_{\mathrm{bal}}^{(m+1)}=\mathrm{F}_{\mathrm{pur}}(F_{\mathrm{bal}}^{(m)},F_{\mathrm{bal}}^{(m)})$. Under symmetric Pauli noise, setting $F'=F$ in Eq.~\eqref{W_pur} gives $\mathrm{F}_{\mathrm{pur}}(F,F)-F=\frac{(1-F)(4F-1)(2F-1)}{5-4F+8F^{2}}$, which is strictly positive for $F\in(1/2,1)$ and vanishes at $F=1/2$ and $F=1$; hence, for any $F^{(0)}>1/2$, $F_{\mathrm{bal}}^{(m)}$ increases monotonically to unity. Under the single-Pauli-error model, the odds $x=(1-F)/F$ of part (i) turn one balanced level into $x\mapsto x^{2}$, so $F_{\mathrm{bal}}^{(m)}=1/(1+x_0^{2^m})\to1$.
\hfill$\blacksquare$

\subsubsection{Diminishing returns and concavity}

For both noise model families, successive pumping rounds satisfy
\begin{equation}
F_{\mathrm{pump}}^{(k+1)}-F_{\mathrm{pump}}^{(k)}
<
F_{\mathrm{pump}}^{(k)}-F_{\mathrm{pump}}^{(k-1)},
\qquad \forall k\geq 1.
\label{eq:concavity}
\end{equation}
For two-Bell states, stemming from the single-Pauli error model, this follows directly from Eq.~\eqref{eq:2b_closed_form}, whose second finite difference is negative since $x_0<1$. 

For the symmetric Pauli noise model (Werner states), Eq.~\eqref{eq:bbpssw_werner} is a M\"obius transformation $h(F)=(aF+b')/(cF+d)$ with $a=10F^{(0)}-1$, $b'=1-F^{(0)}$, $c=8F^{(0)}-2$, and $d=5-2F^{(0)}$. Its determinant is $ad-b'c=42F^{(0)}-12(F^{(0)})^2-3>0$ for $F^{(0)}\in[1/2,1]$ (the identical numerator appearing in Appendix~\ref{app:dp_optimality}), and $c>0$ for $F^{(0)}>1/2$. Since $h'(F)=(ad-b'c)/(cF+d)^{2}$ and $h''(F)=-2c\,(ad-b'c)/(cF+d)^{3}$, this establishes that $h$ is strictly increasing and strictly concave ($h''<0$). To establish Eq.~\eqref{eq:concavity}, note that because $c>0$ the denominator of $h'$ grows with $F$, so for $F\ge 1/2$ we have $cF+d\ge 2F^{(0)}+4$ and therefore $h'(F)\le (ad-b'c)/(2F^{(0)}+4)^{2}$. Moreover,
\begin{equation}
\left(2F^{(0)}+4\right)^{2}-\left(ad-b'c\right)=16\left(F^{(0)}\right)^{2}-26F^{(0)}+19>0,
\end{equation}
since this quadratic has negative discriminant ($=-540$). Thus $h'(F)<1$ on the entangled regime, with $h'(1/2)$ ranging from $3/5$ at $F^{(0)}=1/2$ to $3/4$ at $F^{(0)}=1$, ensuring that $g(F)=h(F)-F$ is strictly decreasing along the increasing sequence $F_{\mathrm{pump}}^{(k)}$, proving Eq.~\eqref{eq:concavity}.

\subsubsection{Consequence}
Successive purification rounds exhibit diminishing fidelity gains, so the attainable $(\mathrm{F},\mathrm{Pr}_{\mathrm{succ}})$ pairs of nested rounds on a single link form a discrete trade-off frontier. The heuristic (Section~\ref{sec:heuristic}) exploits this structure to construct configurations on that frontier at each bottleneck width; the exact formulations do not require it .

\subsection{Property 2. Log-Linearity of Swap Composition}
\label{app:convexity:loglinearity}

For both models, entanglement swapping composes multiplicatively in the Bell  parameter. From Eq.~\eqref{eq:swap_polarization}, the end-to-end Bell parameter of a chain of $L$ links with per-link Bell parameters $\xi_e$ and gate quality $\eta$ is
\begin{equation}
\xi_{\mathrm{e2e}} = \eta^{\,L-1}\prod_{e=1}^{L} \xi_e,
\label{eq:swap_composition}
\end{equation}
which holds for the symmetric Pauli error model with $\xi=\xi_W$ and for the single-Pauli error model with $\xi=\xi_{2B}$.

\subsubsection{Consequence} Taking logarithms,
\begin{equation}
\log \xi_{\mathrm{e2e}} = \sum_{e=1}^{L} \log \xi_e + (L-1)\log\eta,
\label{eq:swap_log}
\end{equation}
which is linear in the per-link log-Bell parameters $\log \xi_e$. The end-to-end fidelity threshold $F_{\mathrm{th}}$ translates into a Bell parameter threshold $\Xi_{\mathrm{th}}$, and the constraint $\xi_{\mathrm{e2e}}\ge\Xi_{\mathrm{th}}$ becomes the linear constraint $\sum_e \log\xi_e \ge \log\Xi_{\mathrm{th}}-(L-1)\log\eta$. Imperfect swap operations, of quality $\eta<1$ (Eq.~\eqref{eq:swap_polarization}), thus shift the right-hand side of $C3$ by the constant $(L-1)\log\eta$ and leave the formulation unchanged; imperfect purification changes only the precomputed coefficients. The MILP of Section~\ref{sec:milp} treats $\log\Xi_{e,k}$ as a precomputed coefficient for each entry of the per-link choice set, which provides the model-dependent input to the fidelity constraint.

\bigskip
\subsection{Property 3. Objective Compatibility}
\label{app:convexity:objective}

The MILP and DP support a class of objectives broader than the throughput maximization evaluated here. The required conditions differ between the two formulations.

\subsubsection{MILP compatibility} The objective is linear in the binary selection variables $\{x_{e,k}\}$, with coefficients determined by the precomputed log-domain quantities $\{\log W_{e,k},\ \log P_{e,k},\ \log \Xi_{e,k}\}$. Property~2 ensures that the relevant end-to-end physical quantities compose multiplicatively and are therefore linearizable under the logarithm; any objective expressible as a sum of such logarithms, together with an auxiliary variable for bottlenecks, preserves the MILP structure.

\subsubsection{DP compatibility} Three conditions are required. i) The objective must \textit{decompose} at every split point of a sub-path, which supplies optimal substructure. ii) Every operation must be \textit{monotone} in each component of the configuration state $(F,W,P)$ separately, which makes componentwise Pareto dominance sufficiency-preserving. iii) The retained set per sub-path must be \textit{finite}; here it is, since each span is built from finite choice sets over finitely many split points, and each post-swap loop ends after at most $\lfloor\log_2 c_{\max}\rfloor$ rounds.

The throughput objective $T=W\cdot\prod_e P_e$ (with end-to-end width $W = \min_e W_e$) used in this work satisfies the conditions of both the MILP and the DP. For the MILP, $\log T = \log W + \sum_e \log P_e$ is linear in the log-domain quantities. For the DP, swapping composes as $(F,W,P)\mapsto\left(F_{\mathrm{swap}},\ \min(W_A,W_B),\ P_AP_B\right)$ and purification as $(F,W,P)\mapsto\left(\mathrm{F}_{\mathrm{pur}},\ \lfloor W/2 \rfloor,\ P^2\cdot\mathrm{Pr}_\mathrm{succ}\right)$; both are nondecreasing in each of $F$, $W$, and $P$ (Lemma~2 of Appendix~\ref{app:dp_optimality}), and $T$ decomposes across any split point $x$ as $T_{ij}=\min(W_{ix},W_{xj})\cdot P_{ix}P_{xj}$. Note that dominance must be tested on $(F,W,P)$ and not on the scalar $T$: because widths compose through a minimum and probabilities through a product, two configurations of equal $T$ are not interchangeable downstream.

\subsubsection{Consequence} Other objectives in this class --- end-to-end success probability, or products $W^{a}P^{b}\xi_{\mathrm{e2e}}^{c}$ that weight fidelity against throughput --- admit the same MILP and DP without structural change, for any noise model satisfying Property~2 and the monotonicity above; multi-request utilities~\cite{kar2026_utility} additionally require coupling the requests' capacities. The two models considered here are instances of that class rather than the extent of it.


\section{Greedy Counter-Example and the Width Sweep}
\label{app:counterexample}

The heuristic of Section~\ref{sec:heuristic} sweeps the bottleneck width and runs a greedy allocation at each fixed width. A three-link instance shows why the sweep is necessary.

Under the single-Pauli-error model (leading to two-Bell-state diagonal states), let $F_{\mathrm{th}}=0.70$, link $e_1$ be high-capacity and low-fidelity ($c_1=100$, $F_1^{(0)}=0.65$) and $e_2,e_3$ be low-capacity and high-fidelity ($c_2=c_3=10$, $F_2^{(0)}=F_3^{(0)}=0.85$). The unpurified bottleneck is $W=10$ and $F_{\mathrm{e2e}}^{(0)}=0.574$, so purification is required.

A single greedy pass increments the link of highest marginal utility $u_e$ of Eq.~\eqref{eq:utility}. It selects $e_1$ twice. At $\mathbf{k}=(2,0,0)$ a third round on $e_1$ has diminishing returns, and $u_1=0.118$ falls below $u_2=u_3=0.148$. The pass switches to $e_2$ and terminates at $\mathbf{k}=(2,1,0)$, with $F_{\mathrm{e2e}}=0.740$, $W=5$, and $T_{\mathrm{e2e}}=1.18$. The MILP selects $\mathbf{k}=(3,0,0)$, meeting the threshold $F_{\mathrm{e2e}}=0.707$ while retaining $W=10$, for $T_{\mathrm{e2e}}=1.94$ --- an improvement of $64\%$.

The cause is the bottleneck. Because $c_1=100$, three rounds on $e_1$ leave $W_{1,3}=25$, still above the bottleneck, and cost no width. A single round on $e_2$ halves $W_2$ to $5$ and carries the bottleneck with it. The utility $u_e$ prices the throughput lost to one round, but cannot distinguish width spent on a link with slack from width spent on the link that sets the bottleneck. Fixing $W=10$ removes the ambiguity: no round may be placed on $e_2$ or $e_3$, since $W_{2,1}=5<10$, and the greedy pass is forced onto $e_1$, recovering $\mathbf{k}=(3,0,0)$. Sweeping $W$ downward from $\min_e c_e$ therefore repairs the greedy choice, at a cost of one factor of $c_{\mathrm{max}}$ in runtime. Any instance in which the fidelity-limiting and capacity-limiting links differ exposes the same divergence. The same instance exhibits the identical gap and repair when the heuristic runs over tree frontiers instead of pumping depths --- as it must, since under the single-Pauli-error noise model schedule shape is irrelevant (Section~\ref{sec:schedules}) and the frontier reduces to the depth set.


\section{Optimality of the DP Framework}
\label{app:dp_optimality}

This appendix proves Theorem~\ref{thm:dp}: instantiated with a per-link choice set, the DP of Section~\ref{sec:dp} returns the operation sequence of maximum expected throughput, subject to the fidelity threshold $F_{\mathrm{th}}$, over the operation set
\begin{equation}
\mathcal{O}=\{\textsc{raw},\ \textsc{purify},\ \textsc{swap}\}
\end{equation}
on a linear chain. Here $\textsc{raw}$ generates an elementary pair; $\textsc{purify}$ on a link is a round of the link's purification schedule --- a tree over raw pairs, of which pumping is the special case enumerated by the pumping choice set (Section~\ref{sec:schedules}); on a virtual link it is self-purification of the link against a copy of itself; and $\textsc{swap}(i,x,j)$ merges two adjacent virtual links. The proof has three parts: a decomposition lemma, which shows that every sequence in $\mathcal{O}$ has the form the DP enumerates; a Pareto-sufficiency lemma, which shows that discarded configurations are never needed; and an induction on sub-path length.

\subsection{Dominance Relation}
\label{app:dp_optimality:dominance}

Recall the configuration tuple $\sigma=(F,W,P,s)$ of Section~\ref{sec:dp}, with expected throughput $T(\sigma)=W\cdot P$. We say $\sigma'$ \emph{dominates} $\sigma$ ($\sigma'\succeq\sigma$), if
\begin{equation}
F(\sigma')\ge F(\sigma),\ W(\sigma')\ge W(\sigma),\ P(\sigma')\ge P(\sigma),
\label{eq:dominance}
\end{equation}
with at least one inequality strict. The Pareto filter retains $\sigma\in\mathcal{C}_{i,j}$ if and only if no other configuration in the set $\mathcal{C}_{i,j}$ dominates it.

Dominance is tested on the three components separately, and not on the scalar $T$. Because $\textsc{swap}$ composes widths through a minimum and probabilities through a product, $T$ is not a sufficient statistic: the configurations $(W,P)=(4,\tfrac14)$ and $(2,\tfrac12)$ have equal $T$, yet swapped against a partner of width $3$ and probability $1$ they yield $T=\tfrac34$ and $T=1$ respectively. Dominance on $(F,T)$ would discard the survivor.

Two immediate consequences of Eq.~\eqref{eq:dominance} are used below. First, $T$ is nondecreasing under $\succeq$. Second, the feasibility predicate $F\ge F_{\mathrm{th}}$ is preserved upward: if $\sigma$ is feasible and $\sigma'\succeq\sigma$, then $\sigma'$ is feasible.

\subsection{Decomposition}
\label{app:dp_optimality:decomposition}

Lemma 1 (Tree decomposition).
\emph{Let $s$ be any sequence in $\mathcal{O}$ producing a virtual link on the span $(i,j)$ with $j>i+1$. Then there exist a split node $x\in(i,j)$, sequences $s_L$ on $(i,x)$ and $s_R$ on $(x,j)$, and an integer $m\ge0$ such that}
\begin{equation}
s=\textsc{purify}^{\,m}\bigl(\textsc{swap}(s_L,s_R)\bigr).
\end{equation}

\begin{proof}
Induct on the number of operations in $s$. The final operation of $s$ produces the virtual link $(i,j)$ and is therefore either a $\textsc{swap}$ at some split $x$, in which case $m=0$, or a $\textsc{purify}$ on $(i,j)$. In the latter case, $\textsc{purify}$ on a virtual link is self-purification, so its two operands are copies of a single configuration on $(i,j)$ produced by a strictly shorter sequence $s''$. Applying the induction hypothesis to $s''$ gives the claim with $m$ incremented by one. Sequences on a link ($j=i+1$) are the per-link schedules of the instantiated choice set $K_e$, and are enumerated directly.
\end{proof}

Lemma~1 is the reason Algorithm~\ref{alg:dp} suffices: for each span it iterates over every split $x$ and every pair of sub-path configurations, applies $\textsc{swap}$, and then applies $\textsc{purify}$ repeatedly while the width permits, retaining each intermediate state.

\subsection{Pareto Sufficiency}
\label{app:dp_optimality:pareto}

Lemma 2 (Pareto sufficiency).
\emph{Let $\sigma'\succeq\sigma$. Then every operation in $\mathcal{O}$ applicable to $\sigma$ is applicable to $\sigma'$, and its result on $\sigma'$ dominates or equals its result on $\sigma$.}

\begin{proof}
The transition maps are
\begin{equation*}
{\small
\begin{aligned}
\textsc{swap}(\sigma,\tau) &= \Bigl(F_{\mathrm{swap}}\bigl(F(\sigma),F(\tau)\bigr),\ \min\bigl(W(\sigma),W(\tau)\bigr),\ P_{\sigma,\tau}\Bigr), \\
\textsc{purify}(\sigma)    &= \Bigl(\mathrm{F}_{\mathrm{pur}}\bigl(F(\sigma),F(\sigma)\bigr),\ \bigl\lfloor W(\sigma)/2 \bigr\rfloor,\ P(\sigma)^2\,\mathrm{Pr}_{\mathrm{succ}}\Bigr),
\end{aligned}}%
\end{equation*}
with $\mathrm{Pr}_{\mathrm{succ}}=\mathrm{Pr}_{\mathrm{succ}}(F(\sigma),F(\sigma))$ and $P_{\sigma,\tau} = P(\sigma)P(\tau)$. It suffices that each output component be nondecreasing in each input component.

\emph{Applicability.} $\textsc{swap}$ is unconditional. $\textsc{purify}$ requires $W\ge2$, and $W(\sigma')\ge W(\sigma)\ge 2$. Because the two operands of a $\textsc{swap}$ occupy disjoint links, no capacity constraint couples them, and no feasibility test arises.

\emph{Fidelity.} By Property~2 of Appendix~\ref{app:convexity}, $\xi_{\mathrm{swap}}=\eta\,\xi(F(\sigma))\,\xi(F(\tau))$, and $\xi$ is an increasing affine function of $F$ in both models; hence $F_{\mathrm{swap}}$ is nondecreasing in $F(\sigma)$ whenever $\xi(F(\tau))\ge0$, which holds on the entangled regime. For purification, $\mathrm{F}_{\mathrm{pur}}$ is symmetric and, in the two-Bell-state diagonal states,
\begin{equation*}
\frac{\partial \mathrm{F}_{\mathrm{pur}}}{\partial F} = \frac{F'(1-F')}{\bigl[FF'+(1-F)(1-F')\bigr]^{2}} \;\ge\; 0,
\end{equation*}
while in the Werner states
\begin{equation*}
\frac{\partial \mathrm{F}_{\mathrm{pur}}}{\partial F} = \frac{-12F'^{2}+42F'-3}{\bigl[5-2F-2F'+8FF'\bigr]^{2}} \;>\; 0
\end{equation*}
for $F'>(7-3\sqrt{5})/4\approx0.073$, and in particular on $F'>1/2$. Restricting to the diagonal $F=F'$ therefore gives an increasing map.

\emph{Width.} $\min(\cdot,W(\tau))$ and $\lfloor\cdot/2\rfloor$ are both nondecreasing.

\emph{Probability.} $P(\sigma)P(\tau)$ is nondecreasing in $P(\sigma)$. For purification, $P^2\cdot\mathrm{Pr}_{\mathrm{succ}}(F,F)$ is nondecreasing in $P$ (since $P \ge 0$), and $\mathrm{Pr}_{\mathrm{succ}}(F,F)$ is increasing in $F$ on the entangled regime, since its derivative is $4F-2$ in the two-Bell-state diagonal states and $(16F-4)/9$ in the Werner states.

Every component is therefore monotone, and $\textsc{swap}(\sigma',\tau)\succeq\textsc{swap}(\sigma,\tau)$ and $\textsc{purify}(\sigma')\succeq\textsc{purify}(\sigma)$.
\end{proof}

Lemma~2 depends only on the monotonicity of the transition maps, not on their particular algebraic form. It therefore holds for any noise model satisfying the conditions of Appendix~\ref{app:convexity}, and for any accounting of the cumulative success probability that is nondecreasing in $P$ and in $F$. The same monotonicity is what makes Pareto filtering of the per-link choice set itself lossless: a frontier point dominated in $(F,W,P)$ can be discarded from $K_e$ without discarding the optimum.

\subsection{Global Optimality}
\label{app:dp_optimality:theorem}

\begin{proof}[Proof of Theorem~\ref{thm:dp}]
We show by induction on the sub-path length $\ell=j-i$ that $\mathcal{C}_{i,j}$ is \emph{sufficient}: for every sequence $s$ in $\mathcal{O}$ producing a virtual link on $(i,j)$, some $\sigma'\in\mathcal{C}_{i,j}$ dominates the configuration produced by $s$.

\emph{Base case ($\ell=1$).} The base pool $\mathcal{C}_{i,i+1}$ of link $e=(i,i+1)$ instantiates the per-link choice set $K_e$: in the pumping class, it is every admissible depth, which by the link-restricted definition of $\mathcal{O}$ exhausts the sequences on $(i,i+1)$; in the tree class, it is the Pareto frontier over all schedules with at most $c_e$ leaves, which is sufficient for the full schedule family by the dominance relation of Eq.~\eqref{eq:dominance}. Applying the Pareto filter preserves sufficiency by Lemma~2.

\emph{Inductive step.} Assume sufficiency for all sub-paths of length below $\ell$, and let $s$ produce a virtual link on a span $(i,j)$ of length $\ell$. By Lemma~1, $s=\textsc{purify}^{\,m}(\textsc{swap}(s_L,s_R))$ for some split $x$. By the induction hypothesis there exist $\sigma_A\in\mathcal{C}_{i,x}$ and $\sigma_B\in\mathcal{C}_{x,j}$ dominating the configurations of $s_L$ and $s_R$. Algorithm~\ref{alg:dp} evaluates the pair $(\sigma_A,\sigma_B)$ at split $x$, applies $\textsc{swap}$, and appends $\textsc{purify}$ rounds while the width permits; by Lemma~2 the state it obtains after $m$ rounds dominates the configuration of $s$, and by applicability it is reached, since dominance preserves $W\ge2$. The Pareto filter then preserves sufficiency at length $\ell$.

\emph{Termination.} At $\ell=L$ the pool $\mathcal{C}_{1,N}$ is sufficient. Let $T(s^*)$ be any feasible sequence of maximum throughput. Some $\sigma'\in\mathcal{C}_{1,N}$ dominates its configuration, so $\sigma'$ is feasible and $T(\sigma')\ge T(s^*)$. The algorithm returns the feasible configuration of $\mathcal{C}_{1,N}$ maximizing $T$, whose throughput is therefore at least $T(s^*)$, and, being attainable, exactly $T(s^*)$.
\end{proof}

\subsection{Scope and Caveats}
\label{app:dp_optimality:caveats}

The optimality claim is scoped as follows.

\textit{(i) Operation set.} Theorem~\ref{thm:dp} applies to $\mathcal{O}$. It assumes uniform pair pools (every configuration replicates an identical routing tree across width $W$), thereby excluding heterogeneous pools where copies follow different schedules. It further excludes cross-chain purification between distinct routing trees, multi-hop bypass swapping, and distillation protocols other than BBPSSW. Under pumping it excludes link-level recurrence, in which two purified pairs of equal depth are purified against each other; in the two-Bell-state diagonal states this exclusion is without loss because $\mathrm{F}_{\mathrm{pur}}(F_{\mathrm{pump}}^{(k)}, F_{\mathrm{pump}}^{(k)}) = F_{\mathrm{pump}}^{(2k+1)}$ at the identical raw-pair cost of $2k+2$; the cumulative success probability is likewise shape-independent, since for any tree with $b$ leaves it telescopes to $P=(1+x_0^{\,b})/(1+x_0)^{\,b}$ with $x_0=(1-F^{(0)})/F^{(0)}$ (merging subtrees with $a$ and $b$ leaves multiplies $P_aP_b$ by $\mathrm{Pr}_{\mathrm{succ}}=(1+x_0^{a+b})/((1+x_0^{a})(1+x_0^{b}))$), so schedule shape is irrelevant to throughput as well. In the Werner states the constructions differ, and admitting them is precisely what the tree instantiation adds.

\textit{(ii) Noise model.} The transition maps must be monotone in $F$, $W$, and $P$, which Appendix~\ref{app:convexity} verifies for both models considered here.

\textit{(iii) Choice-set resolution.} Exactness is relative to the instantiated choice set. In the pumping instantiation the set is complete up to $k_{\max}=20$, which by Proposition~\ref{prop:fixed_points} saturates well before the bound. In the tree instantiation the frontier is constructed with the $\epsilon$-optimal schedule DP of~\cite{chen2024_jsac}, certifying a per-link loss below $10^{-3}$.

\textit{(iv) Topology and Complexity.} The result is exact for linear chains; extending it to general topologies requires a path-selection layer. Theorem~\ref{thm:dp} asserts correctness, not tractability; retained Pareto set sizes admit no a priori bound and are reported empirically in Section~\ref{sec:results}.

\end{document}